\documentclass[aps, prx, reprint,floatfix,superscriptaddress,longbibliography,twocolumn,notitlepage,nofootinbib]{revtex4-2}

\usepackage{graphicx}
\usepackage{subcaption}
\usepackage{float}
\usepackage{tabularx,booktabs} 
\usepackage{tcolorbox}
\usepackage[inline]{enumitem}
\usepackage{amsmath}
\DeclareMathOperator*{\argmax}{arg\,max}
\DeclareMathOperator*{\argmin}{arg\,min}
\usepackage{tikz}
\usetikzlibrary{shapes, arrows.meta, positioning}

\usepackage{amsmath}
\usepackage{amsthm}
\usepackage{bm}
\newtheorem{theorem}{Theorem} 
\newtheorem{lemma}[theorem]{Lemma}
\newtheorem{proposition}[theorem]{Proposition}

\newtheorem{corollary}[theorem]{Corollary}

\theoremstyle{remark}

\theoremstyle{definition}
\newtheorem{definition}[theorem]{Definition}

\usepackage{mathrsfs} 
\usepackage{physics}
\usepackage{mathtools}
\usepackage{mathbbol}
\DeclareSymbolFontAlphabet{\mathbbol}{bbold}
\usepackage{amssymb}
\DeclareSymbolFontAlphabet{\mathbb}{AMSb}
\usepackage{dsfont}

\usepackage[colorlinks=true,urlcolor=blue,citecolor=blue,linkcolor=blue]{hyperref}
\usepackage{cleveref}
\usepackage{orcidlink}

\usepackage[linesnumbered,ruled,vlined]{algorithm2e}
\RestyleAlgo{ruled} 
\usepackage{comment}

\usepackage{array}
\usepackage{booktabs}

\newcommand*\circled[1]{\tikz[baseline=(char.base)]{
            \node[shape=circle,draw,inner sep=2pt] (char) {#1};}}

\newcommand{\red}[1]{\textcolor{red}{#1}}
\newcommand{\blue}[1]{\textcolor{blue}{#1}}

\def\RF{\mathrm{RF}}
\def\Gr{\mathrm{G}}
\def\RG{\mathrm{RG}}
\def\RSG{\mathrm{RSG}}
\def\BPSP{\mathrm{BPSP}}
\def\curry{\mathrm{curry}}
\def\dl{\#_\mathrm{dl}}
\def\Wtot{W_{\mathrm{tot}}}
\def\wt{\mathrm{wt}}

\newcommand{\qinc}{Quantum Innovation Centre (Q.InC), Agency for Science, Technology and Research (A*STAR), 4 Fusionopolis Way, Kinesis \#05-01, Singapore 138635, Republic of Singapore\looseness=-1}

\newcommand{\sutd}{Science, Mathematics and Technology Cluster, Singapore University of Technology and Design, 8 Somapah Road, Singapore 487372, Republic of Singapore\looseness=-1}

\newcommand{\ihpc}{Institute of High Performance Computing (IHPC), Agency for Science, Technology and Research (A*STAR), 1 Fusionopolis Way, \#16-16 Connexis, Singapore 138632, Republic of Singapore\looseness=-1}

\newcommand{\cqct}{Department of Quantum Science and Technology, Research School of Physics, Australian National University, Acton 2601, Australia\looseness=-1}

\newcommand{\cqt}{Centre for Quantum Technologies (CQT), National University of Singapore, Singapore 117543, Republic of Singapore\looseness=-1}

\newcommand{\sit}{Engineering Cluster, Singapore Institute of Technology, 1 Punggol Coast Road, Singapore 828608, Republic of Singapore\looseness=-1}

\newlength{\sqA}
\newlength{\sqC}
\begin{document}

\title{Classical and Quantum Heuristics for the Binary Paint Shop Problem}

\author{V Vijendran\,\orcidlink{0000-0003-3398-1821}}
\email{vjqntm@gmail.com}
\affiliation{\cqt}
\affiliation{\qinc}
\affiliation{\cqct}

\author{Dax Enshan Koh\,\orcidlink{0000-0002-8968-591X}}
\affiliation{\qinc}
\affiliation{\ihpc}
\affiliation{\sutd}
\affiliation{\sit}

\author{Ping Koy Lam\,\orcidlink{0000-0002-4421-601X}}
\affiliation{\cqt}
\affiliation{\qinc}
\affiliation{\cqct}

\author{Syed M Assad\,\orcidlink{0000-0002-5416-7098}}
\affiliation{\qinc}
\affiliation{\cqct}


\begin{abstract}
The Binary Paint Shop Problem (BPSP) is an $\mathsf{APX}$-hard optimisation problem arising in automotive manufacturing: given a sequence of $2n$ cars, comprising $n$ distinct models each appearing twice, the task is to decide which of two colours to paint each car so that the two occurrences of each model are painted differently, while minimising consecutive colour swaps. Performance is measured by the paint swap ratio---the average number of colour changes per car---which directly affects production efficiency and cost. Prior work showed that the Quantum Approximate Optimisation Algorithm (QAOA) at depth $p=7$ achieves a ratio of $0.393$, outperforming the then-best-known classical heuristic, Recursive Greedy, with a proven asymptotic expected ratio of $0.4$ [Phys.\ Rev.\ A 104, 012403 (2021)]. More recently, the Recursive Star Greedy (RSG) heuristic was proposed with a conjectured expected ratio of $0.361$. Here, we formally reduce BPSP to weighted MaxCut, yielding an exact Ising formulation, and benchmark two state-of-the-art low-depth QAOA variants---eXpressive QAOA and Recursive QAOA---at $p=1$ (denoted XQAOA$_1$ and RQAOA$_1$) against four classical BPSP heuristics. Across instances ranging from $2^7$ to $2^{12}$ cars, XQAOA$_1$ achieves an average paint-swap ratio of $0.357$; in our instance-wise benchmarks, it outperforms RQAOA$_1$ on almost all instances and each of the four classical baselines, including RSG, on every instance. Surprisingly, RQAOA$_1$ deteriorates with problem size despite employing provably near-optimal QAOA$_1$ parameters at each recursive step, eventually falling behind RSG at the largest sizes. To our knowledge, this is the first report of size-dependent performance degradation of RQAOA$_1$ at scale. In contrast, XQAOA$_1$ remains stable across all tested sizes, with no evidence of size-dependent degradation.
\end{abstract}

\maketitle

\section{Introduction}
The advent of noisy intermediate-scale quantum (NISQ) devices~\cite{preskill2018quantum,preskill2021quantum,cheng2023noisy} has opened new opportunities toward achieving near-term quantum advantage~\cite{huang2025vast}. Although constrained by limited qubit connectivity, finite coherence times, and the absence of fault-tolerant quantum error correction, such devices can nevertheless execute relatively shallow quantum circuits with reasonable fidelity~\cite{beverland2022assessing,lau2022nisq}. These limitations have spurred the development of algorithms designed to maximise the utility of scarce quantum resources, most notably variational quantum algorithms (VQAs)~\cite{cerezo2021variational,bharti2022noisy,farhi2014quantum,peruzzo2014variational}, which combine quantum state preparation with classical parameter optimisation to adapt to device imperfections while tackling classically hard problems. A prominent example is the Quantum Approximate Optimisation Algorithm (QAOA)~\cite{farhi2014quantum}, which has been extensively studied for combinatorial optimisation~\cite{farhi2014quantum,farhi2015quantum,vikstaal2020applying,harrigan2021quantum,farhi2022quantum,anschuetz2019variational,harwood2021formulating,amaro2022case,azad2022solving,fitzek2022applying,fitzek2024applying,ruan2023quantum} through numerical benchmarks \cite{zhou2020quantum,lotshaw2021empirical,golden2023numerical}, theoretical analyses~\cite{bittel2021training,hadfield2022analytical,wurtz2021maxcut,ng2024analytical,benchasattabuse2025lower}, and numerous algorithmic variants~\cite{bravyi2020obstacles,yu2022quantum,chandarana2022digitized,herrman2022multi,vijendran2023expressive,gaidai2024performance,yoshioka2023fermionic,bartschi2020grover, zhu2022multi,blekos2024review}. Beyond optimisation, QAOA is also computationally universal~\cite{lloyd2018quantum,morales2020universality} and capable of generating distributions that are classically hard to sample from~\cite{farhi2016quantum,dalzell2020how}.

One application of particular interest is the Binary Paint Shop Problem (BPSP)~\cite{epping2001some,epping2004complexity,gupta2013approximability}, an $\mathsf{APX}$-hard optimisation problem arising in automotive manufacturing. The input is a sequence of $2n$ cars comprising $n$ distinct models, each appearing twice. The objective is to paint each car in one of two colours such that the two occurrences of every model receive different colours, while minimising the number of colour transitions between consecutive cars. Performance is quantified by the paint swap ratio---defined as the average number of adjacent pairs of cars painted in different colours---which directly reflects production efficiency and material usage. Streif et al.~\cite{streif2021beating} showed that QAOA at depth $p=7$ (QAOA$_7$) achieves an expected paint swap ratio of $0.393$, surpassing the Recursive Greedy (RG) algorithm~\cite{andres2011some, andres2021greedy}, which has a proven expected ratio of $0.4$ and was, at the time, the best known classical heuristic for the BPSP. More recently, a new classical heuristic---the Recursive Star Greedy (RSG) heuristic~\cite{hanvcl2023improved}---was proposed, with a conjectured expected ratio of $0.361$. While QAOA$_7$’s ability to outperform RG highlights the potential of quantum algorithms for solving complex real-world problems, matching or exceeding the conjectured performance of RSG would require quantum circuits far deeper than current devices can sustain. Indeed, noise accumulates rapidly with circuit depth~\cite{xue2021effects,wang2021noise,marshall2020characterizing,alam2019analysis,alam2020design,streif2021quantum,anschuetz2022quantum,stilck2021limitations}, and while error mitigation can help~\cite{weidinger2023error,shaydulin2021error,botelho2022error,koh2024readout,kandala2019error,zhou2020quantum,koczor2021exponential,leymann2020bitter}, its overhead often outweighs the modest performance gains \cite{takagi2022fundamental,quek2024exponentially}.

Motivated by the limitations of deep QAOA circuits, we investigate two state-of-the-art low-depth variants---eXpressive QAOA (XQAOA)~\cite{vijendran2023expressive} and Recursive QAOA (RQAOA)~\cite{bravyi2020obstacles}---both of which have been reported to outperform the Goemans-Williamson algorithm~\cite{goemans1995improved} on specific families of large MaxCut instances at $p=1$. We begin by establishing a formal reduction of BPSP to the weighted MaxCut problem and subsequently to an Ising spin model. Central to this reduction is the Initial Car Colour (ICC) encoding, which we formalise based on a technique introduced by Streif et al.~\cite{streif2021beating}, that halves the number of decision variables. We further derive exact interaction strengths from BPSP constraints and prove new structural properties of the resulting graphs, including a maximum degree bound and an explicit relation between total graph weight and paint swap cost. Together, these results provide a rigorous foundation for applying QAOA and its variants to BPSP and may be of independent interest for related optimisation problems.

We benchmark RQAOA$_1$ and XQAOA$_1$ on random BPSP instances with up to $2^{12}$ cars, comparing them against four classical heuristics developed for the BPSP: Red-First, Greedy, Recursive Greedy, and Recursive Star Greedy. Our simulations show that XQAOA$_1$ achieves an average swap ratio of $0.357$. In the instance-wise comparisons, it outperforms each of the four classical baselines, including RSG, on every instance and RQAOA$_1$ on almost all instances. In contrast, RQAOA$_1$, despite employing provably near-optimal QAOA$_1$ parameters at each recursive step, exhibits diminishing performance with problem size, falling behind RSG for sufficiently large instances. To our knowledge, this is the first report of RQAOA$_1$’s performance degradation at scale. XQAOA$_1$, by contrast, remains robust across all tested sizes, with no evidence of size-dependent degradation over the range studied.

A related study by Mooney et al.~\cite{mooney2026optimization} also applies RQAOA to the BPSP. Their approach differs from ours in that they avoid direct optimisation of QAOA parameters within each recursion, instead leveraging parameter concentration and parameter transfer to set QAOA parameters using values optimised for other instances. They benchmark RQAOA for $1\leq p\leq 3$ on instances of up to 20 cars and report no noticeable performance loss under parameter transfer. While these results indicate that transferred parameters can sustain RQAOA performance on small instances, our findings reveal that even with fully optimised parameters, RQAOA$_1$ exhibits diminishing returns as problem size increases. Beyond the shared application of RQAOA to the BPSP, the focus and conclusions of the two studies are distinct.

The remainder of this paper is organised as follows. \Cref{prelim} reviews the BPSP, classical heuristics, and relevant quantum optimisation algorithms. \Cref{spinglass_sec} describes the reduction of BPSP to MaxCut and its equivalent Ising formulation. \Cref{num_results_sec} reports our numerical benchmarking results, while \cref{disscussions_sec} analyses their implications and limitations. Finally, \cref{conc_sec} concludes with a summary of insights and directions for future work.

\section{Preliminaries} \label{prelim}

\subsection{The Binary Paint Shop Problem} \label{bpsp_sec}

\begin{figure*}[t]
\includegraphics[width=\textwidth]{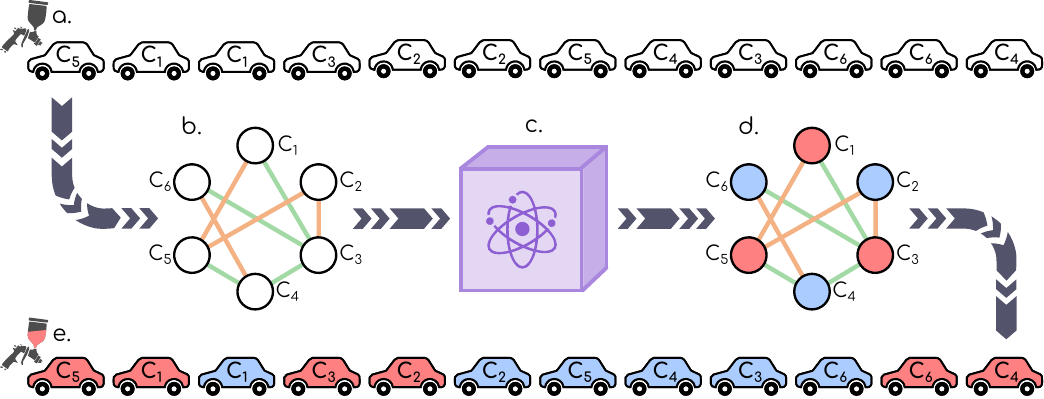}
  \captionsetup{justification=raggedright, singlelinecheck=false}
    \caption{\textbf{Solving the BPSP on a Quantum Computer Using the ICC Encoding.} This figure illustrates the process of solving the BPSP with the Initial-Car-Colour (ICC) encoding using QAOA and its variants. a) The process starts with a BPSP instance, which in this example consists of 6 cars of length 12 forming the sequence C$_5$C$_1$C$_1$C$_3$C$_2$C$_2$C$_5$C$_4$C$_3$C$_6$C$_6$C$_4$. b) Using the ICC encoding the BPSP instance with 6 cars is mapped to an Ising model with 6 spins. In the depicted graphs, orange edges symbolise an edge weight of +1, while green edges indicate an edge weight of -1. c) The Ising model is then solved on a quantum computer using QAOA or its variants. d) This is the ground state solution of the Ising model, where the colour of node (based on the spin) determines the colour of the first occurrence of the car; using this colour the second occurrence of the car is painted the opposite colour. e) This is the solved BPSP instance with a total of 4 paint swaps.}
    \label{car_paints}
\end{figure*}

The binary paint shop problem (BPSP)~\cite{epping2001some,epping2004complexity,gupta2013approximability} is a combinatorial optimisation problem that arises in the paint shops of the automobile industry. A production line receives a fixed sequence of $n$ distinct car models, each appearing exactly twice in the sequence. The order is predetermined by upstream and downstream factory stages—such as the body shop before the paint shop, and trim/chassis/assembly after it—and cannot be rearranged. Each car must be painted in two different colours (e.g., red and blue), but the order in which each car receives its two colours is left to be decided. This assignment, called a \emph{colouring}, determines how colours are applied across the sequence. The objective is to choose a colouring that minimises the number of colour changes between consecutive cars, since every change requires additional setup time, increases paint waste, and reduces throughput. 

Formally, given a set of $n$ distinct cars $C=\left\{c_1, \ldots, c_n\right\}$, and a sequence of car instances $w=(w_1, \ldots, w_{2 n})$ with $w_i \in C$ where each $c \in C$ appears exactly twice, the BPSP asks for a colouring $f=\left(f_1, \ldots, f_{2 n}\right) \in\{r, b\}^{2 n}$ that minimises the total number of colour changes between consecutive cars, subject to the validity constraint that the two occurrences of each car receive opposite colours (i.e., if $w_i=w_j$ and $i \neq j$, then $f_i \neq f_j$). A convenient formulation is:
\begin{equation}
\begin{array}{ll@{}ll}
\text{Minimize}  & \displaystyle\sum\limits_{i=1}^{2n-1} \big[ f_i \neq f_{i+1} \big] &\\
\text{subject to} & f_i \in \{r,b\}, \quad \forall i \in \{1,\ldots,2n\} &\\
                  & f_i \neq f_j, \quad \forall i,j:\ w_i = w_j,\ i \neq j
\end{array}
\label{binary_paint_shop_problem}
\end{equation}
where $[\,\cdot\,]$ denotes the Iverson bracket which evaluates to 1 if its condition holds and 0 otherwise.

The BPSP is $\mathsf{NP}$-complete in its decision form~\cite{epping2004complexity} and $\mathsf{APX}$-hard in its optimisation form~\cite{epping2001some, gupta2013approximability}. In other words, unless $\mathsf{P} = \mathsf{NP}$, no polynomial-time classical algorithm can solve all instances exactly, nor approximate them arbitrarily well via a polynomial-time approximation scheme. Moreover, assuming the Unique Games Conjecture, BPSP is $\mathsf{NP}$-hard to approximate within any constant factor~\cite{gupta2013approximability}. This is in contrast to standard MaxCut instances with non-negative edge weights, which admit polynomial-time constant-factor approximation algorithms such as Goemans-Williamson~\cite{goemans1995improved}. In subsequent sections, we explain that the weighted MaxCut instances arising from BPSP generally contain both positive and negative edge weights, so such guarantees do not directly apply.

Given these complexities, several heuristic approaches have been proposed for the BPSP, each producing solutions with colour changes linear in $n$~\cite{amini2010greedy,andres2021greedy,andres2011some}. The simplest is the red-first algorithm~\cite{andres2011some}, which colours the first occurrence of every car red and the second blue. The greedy algorithm~\cite{amini2010greedy,andres2021greedy} improves on this by selecting at each step the colour that minimises immediate paint changes, switching only when necessary to satisfy the requirement that each car be painted in both colours. The recursive greedy (RG) algorithm~\cite{andres2011some,andres2021greedy} offers a further improvement by recursively removing the last-occurring car from the sequence, optimally colouring the reduced instance, and then reinserting the removed pair in the best possible way. More recently, Hančl et al.~\cite{hanvcl2023improved} proposed the recursive star greedy (RSG) algorithm, a refinement of RG that introduces a temporary ``undecided" colour (star) for car pairs that can be coloured either way without affecting the current swap count; during recursion, if placing a car next to a starred car would otherwise increase the swap count, both occurrences of that starred car are recoloured with binary colours to avoid the increase, and the final colouring is obtained by replacing any remaining stars with arbitrary legal binary assignments. Detailed explanations of these heuristics are provided in \cref{eq:coloring_schemes}. The asymptotic performances of red-first, greedy, and RG are proven, with $\mathbb{E}_{\text{RF}}(\Delta_C) = 2n/3$, $\mathbb{E}_{\text{G}}(\Delta_C) = n/2$, and $\mathbb{E}_{\text{RG}}(\Delta_C) = 2n/5$; in contrast, RSG’s performance is unproven but conjectured to be $\mathbb{E}_{\text{RSG}}(\Delta_C) = 0.361n$, an improvement over RG. Despite these advances, numerical results indicate that the optimal number of colour changes grows linearly in $n$, but with a much smaller constant factor than current heuristics achieve. A recent lower bound of $\mathbb{E}_{\text{OPT}}(\Delta_C) \geq 0.214n - o(n)$ by Hančl et al.~\cite{hanvcl2023improved} shows that the optimal number of colour changes remains linear in $n$. The gap between this bound and the performance of existing heuristics for the BPSP leaves room for improved algorithms, although whether the bound is tight remains open.

\begin{table}[htpb]
\centering
\begin{tabularx}{\columnwidth}{>{\centering\arraybackslash}X >{\centering\arraybackslash}p{4.25cm} >{\centering\arraybackslash}X}
\toprule
\textbf{Heuristic} & \textbf{Solution} & \textbf{\# Swaps} \\
\midrule
Red-First & \red{C$_5$C$_1$}\blue{C$_1$}\red{C$_3$C$_2$}\blue{C$_2$C$_5$}\red{C$_4$}\blue{C$_3$}\red{C$_6$}\blue{C$_6$C$_4$} & 7 \\[0.5em]
Greedy    & \blue{C$_5$C$_1$}\red{C$_1$C$_3$C$_2$}\blue{C$_2$}\red{C$_5$C$_4$}\blue{C$_3$C$_6$}\red{C$_6$}\blue{C$_4$} & 6 \\[0.5em]
RG   & \red{C$_5$C$_1$}\blue{C$_1$C$_3$C$_2$}\red{C$_2$}\blue{C$_5$}\red{C$_4$C$_3$C$_6$}\blue{C$_6$C$_4$} & 5 \\[0.5em]
RSG   & \blue{C$_5$C$_1$}\red{C$_1$}\blue{C$_3$C$_2$}\red{C$_2$C$_5$C$_4$C$_3$C$_6$}\blue{C$_6$C$_4$} & 4 \\
\bottomrule
\end{tabularx}
\captionsetup{justification=raggedright, singlelinecheck=false}
\caption{This table compares the paint sequences and corresponding number of colour swaps for a 6-car BPSP instance, showing the solutions obtained by the classical heuristics together with their swap counts. In this example, red-first produces the highest number of swaps, while RSG achieves the fewest and matches the optimal solution.}
\label{heuristic_comparison}
\end{table}

We conclude this subsection with an illustrative example of the BPSP, shown in \cref{car_paints} with the corresponding results summarised in \cref{heuristic_comparison}. The example compares the four classical heuristics described earlier on the same instance. As expected from their asymptotic behaviour, red-first produces the highest number of swaps, while RSG achieves the fewest and, in this case, matches the optimal solution. It is worth noting that although RG generally outperforms the standard greedy algorithm---as it does in this example---there are instances where the greedy algorithm produces fewer paint swaps~\cite{andres2021greedy}.

\subsection{Quantum Approximate Optimisation Algorithm} \label{qaoa_sec}

The Quantum Approximate Optimisation Algorithm (QAOA)~\cite{farhi2014quantum} is a variational quantum algorithm for approximating solutions to combinatorial optimisation problems. The problem is encoded into a problem Hamiltonian $H_P$, whose ground state represents the optimal solution. QAOA seeks to approximate this ground state by preparing a depth-$p$ variational quantum circuit, parameterised by $2 p$ real angles, and optimising them to minimise the expectation value $\left\langle H_P\right\rangle$. Measuring the resulting quantum state in the computational basis produces candidate bit strings, from which a low-cost solution is selected.

QAOA alternates between two parameterised unitaries: one generated by the problem Hamiltonian $H_P$, which encodes the cost function and biases the state toward low-energy configurations, and one by the mixing Hamiltonian $H_M$, which drives transitions between computational basis states. For an $n$-qubit instance, we take
\begin{equation}
    H_P = \sum_{j=1}^n \sum_{k=1}^n J_{j k} Z_j Z_k, \quad H_M = \sum\limits^n_{j=1} X_j, \label{qising_eqn}
\end{equation}
where $Z_j$ and $X_j$ are Pauli operators acting on qubit $j$, and $J_{jk}$ are problem-specific couplings. 

The corresponding parameterised unitaries are
\begin{align}
    U(H_P, \gamma) &= e^{-i \gamma H_P} = \prod_{j, k}^n e^{-i \gamma J_{jk} Z_j Z_k},\\
    U(H_M, \beta) &= e^{-i \beta H_M} = \prod_{j = 1}^n e^{-i \beta X_j},
\end{align}
with $\gamma \in \mathbb{R}$ and $\beta \in[0, \pi]$.

Starting from the uniform superposition $|+\rangle^{\otimes n}$, the level-$p$ QAOA (QAOA$_p$) is
\begin{equation}
   |\psi(\boldsymbol{\beta}, \boldsymbol{\gamma})\rangle=\prod_{l=1}^p U\left(H_M, \beta_l\right) U\left(H_P, \gamma_l\right)|+\rangle^{\otimes n}. 
\end{equation}
The parameters $(\boldsymbol{\beta}, \boldsymbol{\gamma})$ are optimised classically to obtain
\begin{equation}
\label{eq:QAOAoptparas}
    (\bm{\beta}^*, \bm{\gamma}^*) = \underset{\boldsymbol{\beta}, \boldsymbol{\gamma}}{\text{argmin}} \langle \psi(\boldsymbol{\beta}, \boldsymbol{\gamma}) | H_P | \psi(\boldsymbol{\beta}, \boldsymbol{\gamma}) \rangle  \; .
\end{equation}
The optimised state $\left|\psi\left(\boldsymbol{\beta}^*, \boldsymbol{\gamma}^*\right)\right\rangle$ is then prepared on a quantum computer, measured repeatedly in the computational basis, and the lowest-cost bit string obtained is reported as the solution.

\subsection{Expressive QAOA} \label{xqaoa_sec}

The Expressive QAOA (XQAOA)~\cite{vijendran2023expressive} is an overparameterised extension of standard QAOA, designed to increase ansatz expressiveness at shallow circuit depths. While standard QAOA uses a single $\gamma$ for all terms in $H_P$ and a single $\beta$ for all terms in $H_M, \mathrm{XQAOA}$ assigns independent parameters to each term of both Hamiltonians. This allows different qubits and couplings to be rotated by distinct angles, enabling a richer exploration of the solution space.

XQAOA further augments the standard mixer by introducing an additional $\boldsymbol{\alpha}$-dependent unitary acting on Pauli-$Y$ operators. Let $\boldsymbol{A}=\left(A_j\right)_{j=1}^n$ denote the collection of single-qubit Pauli-$Y$ operators with $A_j=Y_j$, and let $\boldsymbol{B}=\left(B_j\right)_{j=1}^n$ denote the corresponding Pauli-$X$ operators with $B_j=X_j$. For the problem Hamiltonian, we write $\boldsymbol{C}=\left(H_{P_{j k}}\right)_{(j, k) \in E}$, where each $H_{P_{j k}}=J_{j k} Z_j Z_k$ corresponds to the Ising interaction on edge $(j, k)$.

The mixing unitaries with parameter vectors $\boldsymbol{\alpha}=\left(\alpha_1, \ldots, \alpha_n\right)$ and $\boldsymbol{\beta}=\left(\beta_1, \ldots, \beta_n\right)$ are defined as
\begin{align}
    U(\boldsymbol{A}, \boldsymbol{\alpha}) &=e^{-i \sum_{j=1}^n \alpha_j A_j}=\prod_{j=1}^n e^{-i \alpha_j Y_j},\\
    U(\boldsymbol{B}, \boldsymbol{\beta}) &=e^{-i \sum_{j=1}^n \beta_j B_j}=\prod_{j=1}^n e^{-i \beta_j X_j},
\end{align}
where $\alpha_j, \beta_j \in[0, \pi]$ are independent angles acting locally on each qubit. The problem unitary is likewise parameterised per interaction term as
\begin{equation}
    U(\boldsymbol{C}, \boldsymbol{\gamma})  = \prod_{j,k}^n e^{-i \gamma_{jk} J_{jk}Z_jZ_k},
\end{equation}
where $\gamma_{j k} \in \mathbb{R}$ are the parameters associated with each $Z_j Z_k$ term in the problem Hamiltonian $H_P$.

At $p = 1$, XQAOA prepares the state
\begin{align}
    |\psi(\boldsymbol{\alpha}, \boldsymbol{\beta}, \boldsymbol{\gamma})\rangle  &=  U(\boldsymbol{A}, \boldsymbol{\alpha}) U(\boldsymbol{B}, \boldsymbol{\beta}) U(\boldsymbol{C}, \boldsymbol{\gamma}) |+\rangle^{\otimes n} \\
    &=  \prod_j^n e^{-i \alpha_j  Y_j} e^{-i \beta_j  X_j} \prod_{j,k}^n e^{-i \gamma_{jk} J_{jk}Z_jZ_k} |+\rangle^{\otimes n}.\nonumber
\end{align}

The XQAOA ansatz admits four notable configurations, distinguished by the constraints placed on the mixing unitary angles. The Multi-Angle QAOA (MA-QAOA)~\cite{herrman2022multi} is obtained by setting all $\alpha_j = 0$, yielding a multi-angle extension of the standard X-only QAOA mixer. The XY mixer is the most general form, allowing independent $\alpha_j$ and $\beta_j$ for each qubit. The Y mixer restricts the ansatz to $Y$-rotations by setting all $\beta_j = 0$, while the X=Y mixer applies both $X$ and $Y$ rotations with a shared angle on each qubit, enforcing $\alpha_j = \beta_j$. Numerical results indicate that the X=Y mixer outperforms the other three configurations~\cite{vijendran2023expressive}; therefore, all references to XQAOA in this work pertain to the X=Y mixer.

\subsection{Recursive QAOA} \label{rqaoa_sec}

The Recursive QAOA (RQAOA)~\cite{bravyi2020obstacles} is a non-local variant of the QAOA designed to enhance QAOA's performance, especially at shallow circuit depths. RQAOA leverages QAOA as a subroutine to recursively reduce the problem size until it becomes tractable for brute-force methods. Specifically, RQAOA$_1$---RQAOA that uses QAOA$_1$ as its subroutine---has been shown to outperform state-of-the-art classical algorithms for MaxCut and similar problems~\cite{bravyi2020obstacles, bravyi2021classical, bravyi2022hybrid}. Additionally, RQAOA$_1$ provides performance guarantees for particular graph structures~\cite{bravyi2020obstacles, bae2024recursive, bae2024modified} and has been successfully applied to small-scale, real-world problems~\cite{gulbahar2024majority, palackal2023quantum}.

\begin{algorithm}[htpb]
\SetAlgoLined
\SetArgSty{}
\SetKwInput{KwData}{Input}
\SetKwInput{KwResult}{Output}

\KwData{Problem Hamiltonian $C_0$ on $G_0=(V_0,E_0)$ with
$V_0=\{1,\ldots,n\}$, cutoff size $\eta$}
\KwResult{Approximate Ground State of $C_0$}

\For{$t \gets 0$ \KwTo $n-\eta-1$}{
    Prepare QAOA state
    $\left|\psi_t(\boldsymbol{\gamma},\boldsymbol{\beta})\right\rangle$
    for $C_t$
    
    Optimise parameters:
    \vspace{-0.1cm}
    \[
    (\boldsymbol{\gamma}^*,\boldsymbol{\beta}^*)
    \gets
    \underset{\boldsymbol{\gamma},\boldsymbol{\beta}}
    {\operatorname{argmin}}
    \langle
    \psi_t(\boldsymbol{\gamma},\boldsymbol{\beta})
    |C_t|
    \psi_t(\boldsymbol{\gamma},\boldsymbol{\beta})
    \rangle
    \]
    \vspace{-0.25cm}
    
    Compute $M_{ij}$ for all $\{i,j\}\in E_t$:
    \[
    M_{ij}
    \gets
    \langle
    \psi_t(\boldsymbol{\gamma}^*,\boldsymbol{\beta}^*)
    |Z_iZ_j|
    \psi_t(\boldsymbol{\gamma}^*,\boldsymbol{\beta}^*)
    \rangle
    \]
    \vspace{-0.2cm}
    
    Identify the pair
    $\{u,v\}
    \gets
    \underset{\{i,j\}\in E_t}{\operatorname{argmax}}
    |M_{ij}|$
    
    Impose Constraint:
    $Z_v=\operatorname{sign}(M_{uv})Z_u$
    
    Update Variable Set:
    $V_{t+1}\gets V_t\setminus\{v\}$
    
    Update Hamiltonian:
    \vspace{-0.2cm}
    \begin{equation*}
    C_{t+1}
    \gets
    \operatorname{sign}(M_{uv}) \hspace{-0.5cm}
    \sum_{i\in V_{t+1}\setminus\{u\}} \hspace{-0.5cm}
    J_{iv}^{(t)}Z_iZ_u
    + \hspace{-0.4cm}
    \sum_{\{i,j\}\subseteq V_{t+1}} \hspace{-0.5cm}
    J_{ij}^{(t)}Z_iZ_j
    \end{equation*}
    \vspace{-0.35cm}
    
    with absent couplings zero; constants omitted.
}

Solve $C_{n-\eta}$ exactly by brute force

Solve $C_0$ by backtracking through constraints

\Return{Approximate Ground State of $C_0$}
\caption{RQAOA for Ising Models}
\label{rqaoa_procedure}
\end{algorithm}

At each recursive step, RQAOA executes QAOA on the current problem Hamiltonian $H_P$, calculating the correlation metrics $M_{i j}=\left\langle Z_i Z_j\right\rangle$ for all pairs of qubits $\{i, j\}$. The pair with the highest absolute correlation, identified as $\{u, v\}=\operatorname{argmax}_{\{i, j\}}\left|M_{i j}\right|$, is selected. If $M_{u v}>0$, the qubits $u$ and $v$ are positively correlated, and the constraint $Z_v=Z_u$ is imposed. Conversely, if $M_{u v}<0$, the qubits are anti-correlated, leading to the constraint $Z_v=-Z_u$. These constraints effectively eliminate one qubit by expressing its state in terms of the other, thereby reducing the problem size. This process is iteratively applied to the updated Hamiltonian until the problem size is sufficiently small to be solved exactly using brute-force methods. The solution to the reduced instance is then propagated back through the imposed constraints to derive the solution to the original problem. The detailed steps of RQAOA are outlined in algorithm \ref{rqaoa_procedure}.

\section{Recasting BPSP as an Ising Model} \label{spinglass_sec}

We now present a complete and provably equivalent mapping of the BPSP to an Ising model suitable for quantum optimisation. Our approach formalises the Initial-Car-Colour (ICC) encoding, inspired by Streif et al.~\cite{streif2021beating}, which reduces the number of decision variables by half. Using this encoding we construct, for any BPSP instance, a weighted graph whose MaxCut value directly corresponds to the optimal paint swap cost. In particular, we establish a polynomial-time reduction from BPSP to weighted MaxCut, which, in turn, provides an exact Ising Hamiltonian whose ground state encodes the optimal solution. The full technical development, with proofs and derivations, is given in the appendix; here we present only the essential ingredients needed for the remainder of the paper.

\subsection{Initial Car Colour Encoding}

A BPSP instance is a sequence $x \in \{1,\ldots,n\}^{2n}$ in which every symbol occurs exactly twice. A colouring of $x$ is given by a string $f \in\{r, b\}^{2 n}$, where $r$ (red) and $b$ (blue) denote the two paint colours. The colouring $f$ is valid if, whenever two positions $i$ and $j$ in the sequence correspond to the same car (i.e. $x_i = x_j$), the assigned colours differ ($f_i \neq f_j$). This direct encoding, however, contains redundant information: once the colour of the first occurrence of each car is fixed, the second occurrence must necessarily take the opposite colour. In other words, if the first appearance of car $j$ is assigned colour $f_i$, then its second appearance is automatically $\neg f_i$, where $\neg$ flips red to blue and vice versa. This motivates the Initial-Car-Colour (ICC) encoding, which records only the colours of the first occurrences. Thus, every valid colouring can be described succinctly by an $n$-bit string instead of a $2 n$-bit one. For a rigorous treatment of this encoding, we refer the reader to \cref{sec:icc_encoding_scheme}.

To formalise this compression, we define the expansion map $\mathcal{E}_n$. Given a BPSP instance $x$ and an ICC colouring $z \in\{r, b\}^n$, the corresponding full colouring is defined as $f= \mathcal{E}_n(x, z) \in\{r, b\}^{2 n}$, with entries
\begin{equation}
    \mathcal{E}_n(x, z)_i =\neg^{\left[x_i \in\left\{x_1, \ldots, x_{i-1}\right\}\right]} z_{x_i} \quad \forall i \in [2n],
\end{equation}
where $[\,\cdot\,]$ denotes the Iverson bracket which evaluates to 1 if its condition holds and 0 otherwise. In words, $\mathcal{E}_n$ assigns to each position the colour of that car's first occurrence, flipped if the car has already appeared earlier in the sequence. For example, for the BPSP instance $x=(1,2,1,3,3,2)$ and an ICC string $z=rrb$, the expansion $\mathcal E_3(x,z)$ yields the valid BPSP colouring $rrbbrb$. The mapping from ICC strings to valid colourings is bijective (see proposition~\ref{icc_bijective_prop}), ensuring no information is lost. 

With the ICC encoding and the expansion map $\mathcal{E}_n$, the BPSP objective can be expressed directly in terms of ICC variables $z \in\{r, b\}^n$ (see proposition~\ref{icc_cost_prop}). Given a BPSP instance $x$ and an ICC colouring $z$, the number of paint swaps is
\begin{equation}
    \tilde{\xi}_n(x, z)=\sum_{i=1}^{2 n-1}\left[\neg^{\eta(x, i)} z_{x_i} \neq z_{x_{i+1}}\right],
    \label{icc_obj}
\end{equation}
where the auxiliary function $\eta(x, i)$ determines whether the pair $\left(x_i, x_{i+1}\right)$ contributes a ``same-colour" or ``opposite-colour" preference. It is defined by
\begin{equation}
    \eta(x,i) \! = \! [x_{i+1} \! \in \! \{x_1,\ldots, x_i\} ] \! \oplus \! [x_i \! \in \! \{x_1,\ldots, x_{i-1}\} ],
\end{equation}
with $\oplus$ denoting the exclusive-or operation. In words, $\eta(x, i)=1$ when the boundary $(x_i, x_{i+1})$ consists of one first occurrence and one second occurrence, and $\eta(x, i)=0$ when both are of the same type. Thus, $\eta$ encodes the local parity between neighbours, enabling the swap count to be expressed as a quadratic form in the ICC variables and setting the stage for the graph and Ising formulations.

\subsection{BPSP Graph Construction}

We associate each BPSP instance with a corresponding BPSP graph whose vertices correspond to cars, edges connect cars that appear consecutively in the sequence, and weights reflect the cost preference for colour assignments.

Given a BPSP instance $x=\left(x_1, x_2, \ldots, x_{2 n}\right)$, we associate to it a weighted graph $G_x=$ ($V_x, E_x, W_x$). The vertex set is $V_x=\{1, \ldots, n\}$, corresponding to the cars in the alphabet $[n]$. An edge $\{u, v\}$ is included in $E_x$ whenever the symbols $u$ and $v$ appear consecutively somewhere in the word $x$, provided that $u \neq v$ and the resulting weight is nonzero. The weight of an edge $e \in E_x$ aggregates the contributions of all such adjacencies in the word and is defined as
\begin{equation}
    W_x(e)=-\sum_{k=1}^{2 n-1}(-1)^{\eta(x, k)} \delta_{e,\left\{x_k, x_{k+1}\right\}},
\end{equation}
where $\delta_{e,\left\{x_k, x_{k+1}\right\}}$ equals 1 if the boundary ($x_k, x_{k+1}$) corresponds to the unordered pair $e$ and 0 otherwise. In this way, each edge weight records how many times the pair of cars appears as neighbours in the sequence, corrected by the parity term $(-1)^{\eta(x, k)}$ that enforces whether equal or opposite colours are favoured.

The BPSP graph has two particularly useful structural properties. First,
we show in proposition~\ref{max_deg_prop} that its maximum degree is at
most four. Each car appears twice in the BPSP sequence, and each
occurrence has at most two adjacent positions; hence, a car can have at
most four distinct neighbours in the corresponding graph. Second, as
shown by Streif et al.~\cite{streif2021beating}, the edge weights take
values in $\{\pm1,-2\}$. As $n\rightarrow\infty$, edges of weight $-2$
vanish in probability, while the remaining couplings have limiting
probabilities $2/3$ for weight $-1$ and $1/3$ for weight $+1$.

Beyond local structure, we also identify global properties of the BPSP graph. In proposition~\ref{prop:total_weight_BPSP} we show that its total weight can be expressed compactly in terms of the $\eta$-function and the number of double letters. Corollary~\ref{cor:red_first_cost_total_weight} then relates this quantity directly to the cost of the red-first colouring, revealing that even this simple heuristic is governed by a global invariant of the graph.

Together, these local and global properties establish the BPSP graph as the central combinatorial object underlying the problem. In the next subsection, we exploit this structure to obtain an exact reduction of BPSP to weighted MaxCut.

\subsection{Reduction to Weighted MaxCut}

Once the BPSP graph is defined, the connection to MaxCut becomes clear. MaxCut asks for a partition of the vertices of a weighted graph into two sets such that the total weight of the edges between them is maximised. For a given instance $x$, each ICC assignment $z \in\{r, b\}^n$ induces such a partition of the vertices of $G_x$ : two vertices lie on opposite sides of the cut whenever their ICC colours differ. The weight of this cut is obtained by summing the weights of the edges it separates, making the cut value a direct function of the ICC assignment.

In corollary~\ref{cor:cut_weight_eta} we show that for a BPSP instance $x$ with graph $G_x$, the cut weight induced by an ICC assignment $z$ can be expressed as
\begin{equation}
    \mathrm{wt}_{G_x}(z)=-\sum_{k=1}^{2 n-1}(-1)^{\eta(x, k)}\left[z_{x_k} \neq z_{x_{k+1}}\right].
    \label{cut_obj}
\end{equation}
Building on this, proposition~\ref{cor:paint_swap=red-first_minus_cut weight} establishes a direct link between cut weight and paint swaps:
\begin{equation}
    \tilde{\xi}_n(x, z)=\xi_n\left(\sigma_{\RF}(x)\right)-\wt_{G_x}(z),
    \label{swaps_rf_cut}
\end{equation}
where $\sigma_{\mathrm{RF}}$ denotes the red-first colouring $f \in\{r, b\}^{2 n}$ and $\xi_n$ counts the number of swaps in such colouring. Thus, the BPSP cost can be written as the red-first baseline minus the cut weight achieved by $z$. Using corollary~\ref{cor:red_first_cost_total_weight}, this relation can also be expressed in terms of global graph invariants such as the number of double letters and the total weight, though we omit the details here for brevity.

This reformulation shows that minimising $\tilde{\xi}_n(x, z)$ is exactly equivalent to maximising $\mathrm{wt}_{G_x}(z)$, i.e. solving weighted MaxCut on the BPSP graph. In corollary~\ref{cor:BPSP_reduces_to_MaxCut} we rigorously prove this equivalence and, in algorithm~\ref{alg:reduction_BPSP_MaxCut}, present the polynomial-time reduction from BPSP to the weighted MaxCut problem.

\subsection{Ising Model Formulation}

The weighted MaxCut objective can be naturally reformulated as the ground-state energy of an Ising Hamiltonian. Using spin variables $z \in\{-1,1\}^n$, the cut weight of the BPSP graph can be written as
\begin{equation}
    \mathrm{wt}_{G_x}(z)= \frac{-1}{2} \sum_{k=1}^{2 n-1}(-1)^{\eta(x, k)}\left(1 - z_{x_k}z_{x_{k+1}}\right).
    \label{cut_obj_spin}
\end{equation}
Substituting this expression into \cref{swaps_rf_cut} and replacing the spin variables by Pauli-$Z$ operators yields a quantum Ising Hamiltonian whose ground state encodes the optimal BPSP colouring.

While this already provides a valid mapping, it is possible to obtain a more direct formulation that avoids referencing the red-first cost. By working directly with the ICC objective defined in \cref{icc_obj}, proposition~\ref{prop:Ising1} shows that for a BPSP instance $x$ and ICC assignment $z \in \{-1, 1\}^n$, the number of paint swaps can be expressed as
\begin{equation} 
    \tilde{\xi}_n(x, z)=n-\frac{1}{2}-\frac{1}{2} \sum_{i=1}^{2 n-1}(-1)^{\eta(x, i)} z_{x_i} z_{x_{i+1}}.
\end{equation}
In corollary~\ref{cor:Ising2} we further show that the ground state of this Hamiltonian corresponds exactly to the optimal BPSP colouring, i.e., the colouring with the minimal number of paint swaps.

Hence, solving BPSP with quantum heuristics such as QAOA amounts to preparing the ground state of the quantum Ising Hamiltonian
\begin{equation} 
    H_P(x)=n-\frac{1}{2}-\frac{1}{2} \sum_{i=1}^{2 n-1}(-1)^{\eta(x, i)} Z_{x_i} Z_{x_{i+1}}.
\end{equation}
where the spin variables have been promoted to Pauli-$Z$ operators.

\begin{figure*}[htpb]
  \centering
  \subfloat[Paint Swaps for Varying Heuristics]{%
    \includegraphics[width=\sqA,height=\sqA]{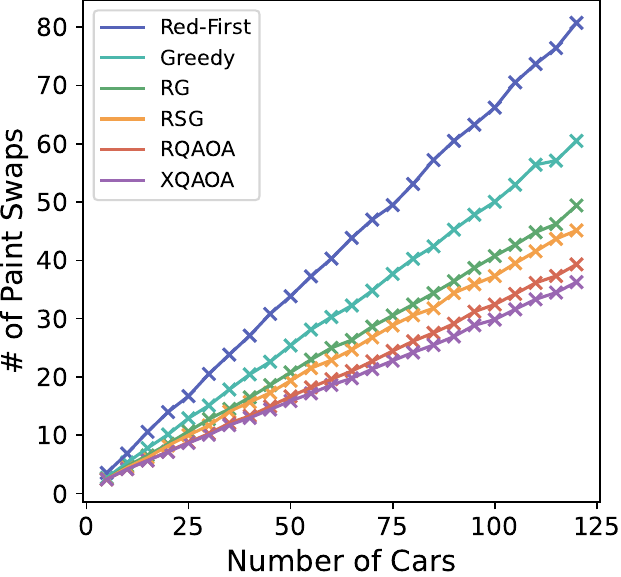}%
    \label{fig:example_a}%
  }\hfill
  \subfloat[Paint Swap Ratios for Varying Heuristics]{%
    \includegraphics[width=0.65\textwidth,height=\sqA]{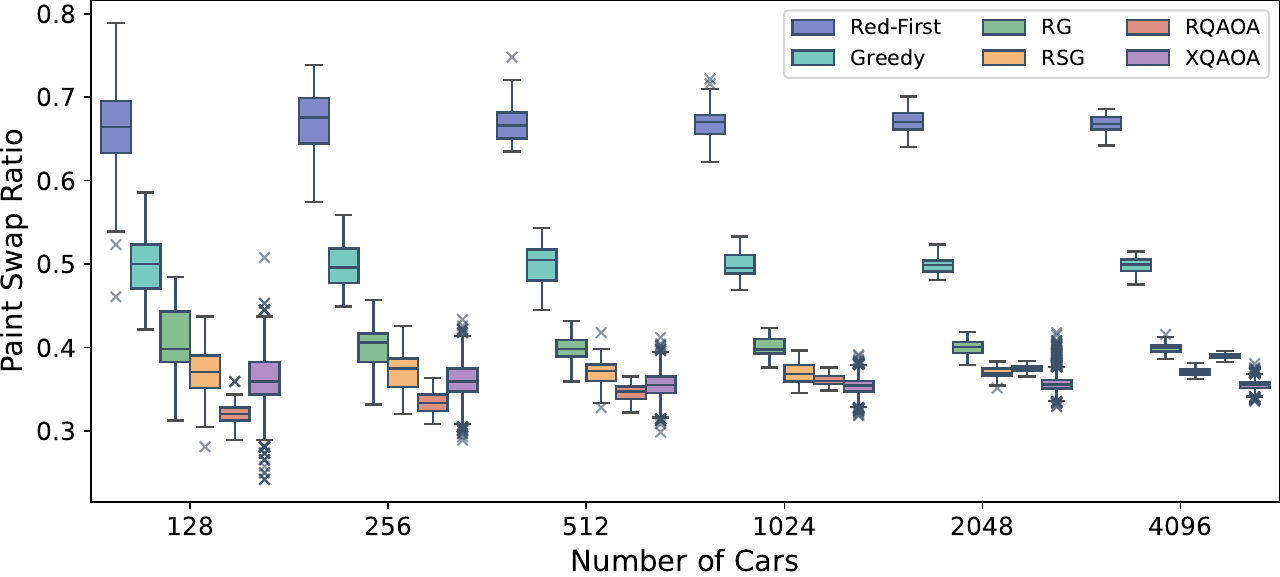}%
    \label{fig:example_b}%
  }\\[1ex]
  \subfloat[Swap Ratios for RSG vs RQAOA]{%
    \includegraphics[width=\sqA,height=\sqA]{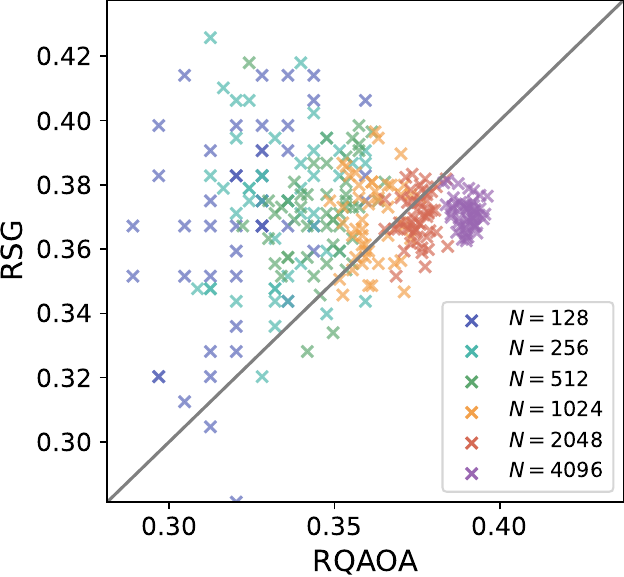}%
    \label{fig:example_c}%
  }\hfill
  \subfloat[Swap Ratios for RSG vs XQAOA]{%
    \includegraphics[width=\sqA,height=\sqA]{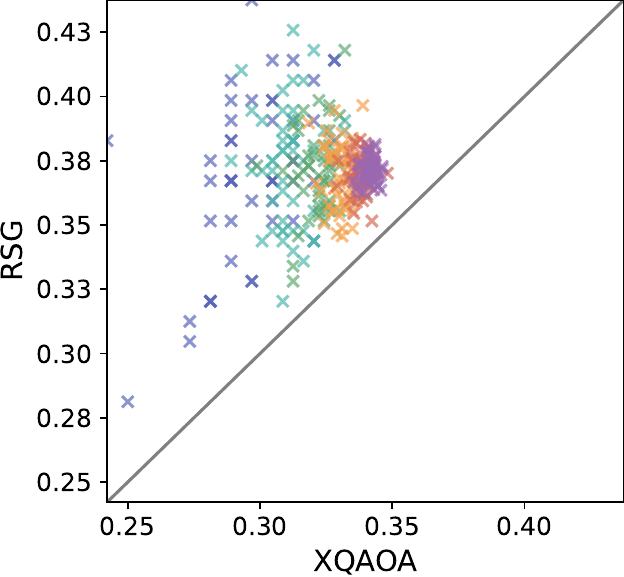}%
    \label{fig:example_d}%
  }\hfill
  \subfloat[Swap Ratios for RQAOA vs XQAOA]{%
    \includegraphics[width=\sqA,height=\sqA]{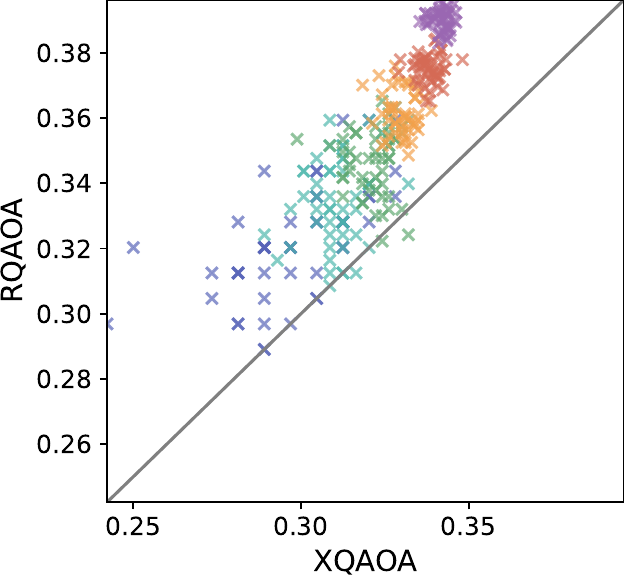}%
    \label{fig:example_e}%
  }
  \captionsetup{justification=raggedright,singlelinecheck=false}
  \caption{\textbf{Benchmarking Classical and Quantum Heuristics on BPSP Instances.}
    (a) Average number of paint swaps obtained by the different heuristics over
    200 random BPSP instances for each
    $n\in\{5,10,\ldots,120\}$. For XQAOA$_1$, each instance is represented by
    the best solution obtained from 100 random initialisations. Among the
    classical heuristics considered here, Red-First performs worst and RSG
    performs best, while both quantum heuristics achieve lower average swap
    counts, with XQAOA$_1$ outperforming RQAOA$_1$.
    (b) Distribution of paint-swap ratios for 50 random BPSP instances at each
    $n\in\{128,256,512,1024,2048,4096\}$. For XQAOA$_1$, the boxplots contain
    all 100 optimised solutions for each instance. Among the classical
    heuristics, RSG has the lowest average ratio, approximately $0.37$, while
    XQAOA$_1$ has an average ratio of approximately $0.357$. RQAOA$_1$
    performs well at smaller sizes, but its performance deteriorates as $n$
    increases.
    (c--e) Instance-wise comparisons between RSG, RQAOA$_1$, and XQAOA$_1$.
    For XQAOA$_1$, these panels use the best solution obtained from 100 random
    initialisations for each instance. In each panel, the diagonal $y=x$
    marks equal performance; because lower paint-swap ratios are better, points
    above the diagonal favour the algorithm on the $x$-axis, while points below
    favour the algorithm on the $y$-axis. Panel (c) shows that RQAOA$_1$
    outperforms RSG on most instances up to $n=1024$, after which RSG becomes
    better, particularly for $n\geq2048$. Panel (d) shows that XQAOA$_1$
    outperforms RSG on every instance across all tested sizes. Panel (e) shows
    that XQAOA$_1$ outperforms RQAOA$_1$ on almost all instances, with only a
    few exceptions.}
    \label{bpsp_finite}
\end{figure*}

\begin{figure*}[htpb]
    \centering
    \subfloat[Swap Ratios for RQAOA$_1$]{
        \includegraphics[height=0.30\textheight]{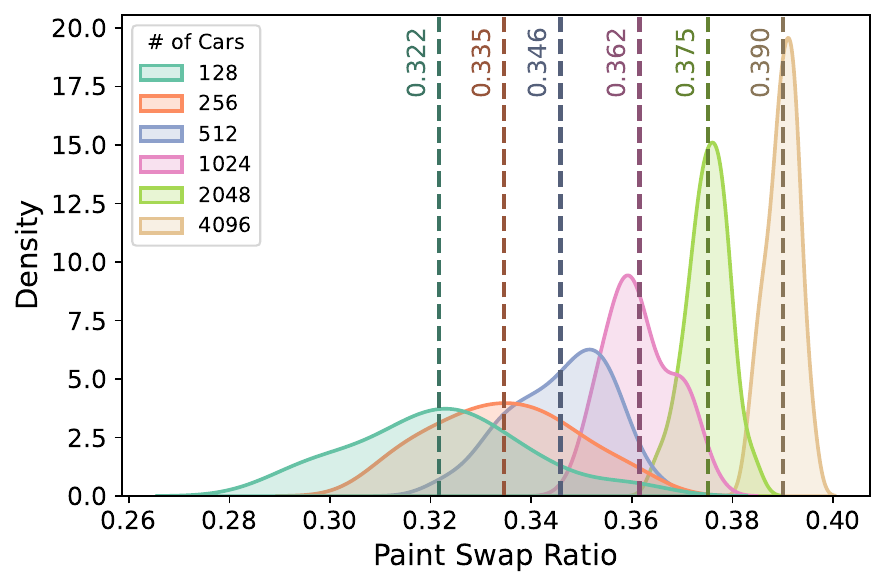}
        \label{rqaoa_swap}
    }
    \subfloat[Swap Ratios for XQAOA$_1$]{
        \includegraphics[height=0.30\textheight]{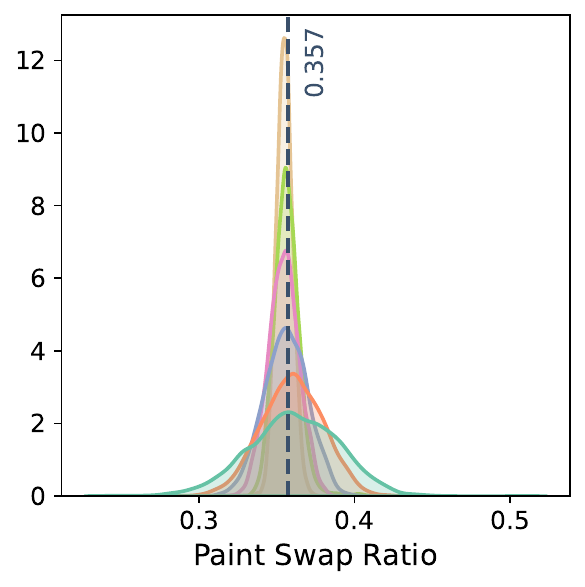}
        \label{xqaoa_swap}
    }
    \captionsetup{justification=raggedright, singlelinecheck=false}
    \caption{\textbf{Distribution of Paint Swap Ratios for RQAOA$_1$ and XQAOA$_1$.} This figure shows the kernel density estimate (KDE) plots of paint-swap ratio distributions for larger instances comprising 50 random problems for each $n \in \{128,256,512,1024,2048,4096\}$. Panel (a) presents RQAOA$_1$ and panel (b) XQAOA$_1$, with the latter including all solutions from 100 restarts per instance. For RQAOA$_1$, the mean of each distribution shifts steadily right as $n$ increases, indicating deteriorating performance: rising from 0.322 at $n=128$ to 0.362 at $n=1024$, where it already falls behind XQAOA$_1$’s average of 0.357. At $n=2048$, the mean increases further to 0.375, worse than RSG’s 0.37, and by $n=4096$ it reaches 0.39, with the right tail extending to 0.4, matching the performance of RG. In contrast, XQAOA$_1$ exhibits stable behaviour across all sizes, with means consistently near 0.357. As $n$ grows over the tested range, its distributions sharpen while remaining concentrated around this value.}
    \label{bpsp_dist}
\end{figure*}

\section{Numerical Results} \label{num_results_sec}

In this section, we benchmark the performance of XQAOA$_1$ and RQAOA$_1$ on the Binary Paint Shop Problem, comparing them against four classical heuristics developed for the BPSP: Red-First, Greedy, Recursive Greedy, and Recursive Star Greedy. All BPSP instances were generated randomly using a quantum random number generator~\cite{symul2011real, haw2015maximization}. Since both XQAOA$_1$ and QAOA$_1$ admit closed-form expressions for their expectation values—given respectively by Theorem 3, Eq. (22), and Theorem 1, Eq. (13) of~\cite{vijendran2023expressive}—all benchmarks were performed entirely on a classical computer, without the use of quantum hardware. For XQAOA$_1$, the large number of parameters motivated the use of the JAX framework for efficient gradient computation via backpropagation, with optimisation carried out using the Limited-Memory Broyden–Fletcher–Goldfarb–Shanno (LBFGS) algorithm~\cite{liu1989limited}. To address the highly non-convex optimisation landscape, each instance was run with 100 random initialisations. For RQAOA$_1$, the pairwise expectation values required at each recursion step were computed analytically as functions of $(\beta,\gamma)$, and the variational parameters were set using the provably near-optimal values obtained by the method of Vijendran et al.~\cite{vijendran2025near}\footnote{In brief, this method applies to Ising models with arbitrary rational weights, reducing the two-dimensional optimisation over $(\beta,\gamma)$ to a one-dimensional search in $\gamma$, using the Nyquist--Shannon theorem to set the sampling resolution, and a subdivision algorithm to locate the optimal angle in linear time.
}, eliminating the need for multiple restarts; hence, each instance was simulated only once. The deterministic classical heuristics were likewise run once per instance.

We consider two benchmark sets. The first comprises 200 random instances for each $n \in \{5, 10, \ldots, 120\}$, probing how the number of paint swaps scales with problem size. The second consists of 50 random instances for each $n \in \{128, 256, 512, 1024, 2048, 4096\}$, used to compare expected swap ratios $\mathbb{E}(\Delta_C/n)$ on larger instances. Results are summarised in \cref{bpsp_finite} and \cref{bpsp_dist}. For XQAOA$_1$, \cref{fig:example_a} and the pairwise comparisons in \cref{fig:example_c}, \cref{fig:example_d}, and \cref{fig:example_e} use the best solution obtained from 100 random initialisations per instance, whereas \cref{fig:example_b} and \cref{xqaoa_swap} show the full distribution over all 100 optimised solutions. The reported average XQAOA$_1$ paint-swap ratio of approximately $0.357$ refers to this full distribution.

\Cref{fig:example_a} shows the average number of paint swaps on smaller instances. Among classical heuristics, Red-First scales the worst, while Recursive Star Greedy is the most effective. Both quantum heuristics outperform the classical heuristics, with XQAOA$_1$ consistently below RQAOA$_1$ and showing the most favourable scaling. The widening separation between methods as $n$ increases highlights how even small improvements in swap ratio compound into large advantages at larger problem sizes.

\Cref{fig:example_b} compares paint-swap ratios for larger instances. The classical heuristics behave in line with theoretical expectations: Red-First, Greedy, and Recursive Greedy converge to asymptotic averages near $0.66$, $0.5$, and $0.4$ respectively, while the Recursive Star Greedy achieves an average of $0.37$, slightly higher than its conjectured value of $0.361$. XQAOA$_1$ attains an average ratio of about $0.357$, outperforming all four classical baselines. In contrast, RQAOA$_1$ exhibits size-dependent performance: it is competitive on smaller instances, but its average ratio rises steadily with $n$, eventually falling behind both XQAOA$_1$ and Recursive Star Greedy. This deterioration is evident in \cref{bpsp_dist}, where the mean of the RQAOA$_1$ distributions drifts steadily rightward, increasing from 0.322 at $n=128$ to 0.335 at $n=256$, 0.346 at $n=512$, and 0.362 at $n=1024$, already worse than the stable XQAOA$_1$ average of 0.357; by $n=2048$ the mean reaches 0.375, exceeding the Recursive Star Greedy benchmark of 0.37, and at $n=4096$ it rises further to 0.390, with the right tail extending to 0.4, matching the performance of Recursive Greedy. Despite relying on provably near-optimal parameters for QAOA$_1$ at every recursive step, RQAOA$_1$ deteriorates consistently with system size. In contrast, XQAOA$_1$ remains strikingly stable, with its distributions sharply peaked around $0.357$ across all tested sizes and becoming increasingly concentrated as $n$ grows.

Returning to \cref{bpsp_finite}, the boxplots in \cref{fig:example_b} compress as variance shrinks with increasing size, making it difficult to distinguish the relative performance of Recursive Star Greedy, RQAOA$_1$, and XQAOA$_1$. To address this, \cref{fig:example_c,fig:example_d,fig:example_e} present pairwise scatter plots, where each point compares the swap ratios of two algorithms and the diagonal $y=x$ line is shown for reference. Points lying above the diagonal indicate that the algorithm on the $x$-axis achieves a lower swap ratio, while points below the diagonal indicate that the $y$-axis algorithm performs better. In \cref{fig:example_c}, RQAOA$_1$ outperforms Recursive Star Greedy up to $n=1024$, but its advantage disappears at $n=2048$ and reverses completely by $n=4096$, where Recursive Star Greedy dominates all instances. By contrast, \cref{fig:example_d} shows that XQAOA$_1$ consistently surpasses Recursive Star Greedy at every size, typically by a substantial margin. Finally, \cref{fig:example_e} reveals that XQAOA$_1$ also dominates RQAOA$_1$ across nearly all instances, with only a handful of small-$n$ cases lying close to the diagonal. Taken together, these comparisons reinforce the trend that XQAOA$_1$ maintains robust performance at scale, while RQAOA$_1$ progressively deteriorates and is ultimately overtaken by RSG.

Across all simulations, Recursive Star Greedy averages about $0.37$, above its conjectured value of $0.361$. RQAOA$_1$ deteriorates with problem size despite using provably near-optimal QAOA$_1$ parameters. In contrast, the distribution of XQAOA$_1$ solutions across all 100 restarts remains
stable around $0.357$, while its best-of-100 solutions outperform each of
the four classical baselines on every instance and RQAOA$_1$ on almost
all instances in the instance-wise comparisons. Together, these results indicate that XQAOA$_1$ sustains competitive performance at scale, whereas RQAOA$_1$ does not.

\section{Discussion} \label{disscussions_sec}

In this section we discuss the key patterns that emerge from our benchmarks. We highlight the diminishing returns observed for RQAOA$_1$ as problem size increases, contrast this with the stable performance of XQAOA$_1$ around a paint-swap ratio of roughly 0.357, and consider what these findings imply for scalability and the choice of algorithms.

\subsection{Diminishing Returns of Level-1 RQAOA}

Our large-scale benchmarks reveal a surprising and previously unreported trend: the performance of RQAOA$_1$ deteriorates with problem size even when the QAOA$_1$ subroutine is run with provably near-optimal parameters~\cite{vijendran2025near}. This stands in sharp contrast to prior guarantees showing that RQAOA$_1$ can exactly solve unweighted cycles~\cite{bravyi2020obstacles}, complete graphs~\cite{bae2024recursive}, and weighted bipartite graphs~\cite{bae2024modified}, as well as large-scale simulations on MaxCut and Ising models without external fields, which reported stable behaviour up to hundreds of qubits with no indication of degradation~\cite{bravyi2020obstacles,vijendran2025near}. Empirically, however, our kernel density estimates in \cref{rqaoa_swap} shift monotonically toward higher paint-swap ratios as $n$ grows, providing the first clear evidence of RQAOA$_1$ deteriorating at scale. Related work~\cite{vijendran2025near} previously reported erratic performance for Ising models with external fields, where RQAOA$_1$ performed very well on some instances but poorly on others; however, those results indicated instance-dependent instability rather than the clear size-dependent deterioration observed here.

This trend suggests that the recursive structure of RQAOA$_1$ may play an important role in the observed degradation, rather than the effect arising from a generic deterioration of the underlying QAOA$_1$ objective with system size. For fixed depth and fixed angles, QAOA is not generally expected to deteriorate simply because the system size increases within a structured problem family. In analytically tractable settings this can be shown explicitly: for unweighted triangle-free regular MaxCut, the QAOA$_1$ expectation value per edge, the optimal angles, and the resulting expected approximation ratio depend on local graph parameters such as the degree, but not on the number of vertices~\cite{hadfield2018quantum,wang2018quantum}. This approximate scale invariance is also one of the motivations behind parameter concentration and transfer across problem sizes~\cite{brandao2018fixed,akshay2021parameter,wurtz2021fixed,shaydulin2023parameter}.

A plausible mechanism for the degradation observed here instead arises from the interplay between the locality of the $p=1$ correlation estimates and the recursive elimination rule. At each recursion step, RQAOA$_1$ selects a pair of variables using the correlations $\langle Z_i Z_j\rangle$ produced by a shallow QAOA$_1$ circuit and then irreversibly imposes either a correlation or an anti-correlation constraint. At $p=1$, these correlation estimates depend only on a bounded local neighbourhood and therefore need not contain sufficient information about the global structure of the remaining instance. An inaccurate choice of pair or correlation sign can consequently introduce a suboptimal constraint. The recursive procedure then propagates that constraint through all subsequent reductions. As the instance size grows, the number of elimination steps also grows, providing more opportunities for such local errors to occur and accumulate into globally suboptimal assignments. In this picture, the locality of the $p=1$ correlation estimates and the recursive elimination are not competing explanations: the former can generate imperfect local decisions, while the latter makes those decisions persistent and allows their effects to compound.

We emphasise, however, that the present simulations do not separately measure these two ingredients and therefore do not constitute a causal isolation of this mechanism. A matched non-recursive baseline or a comparison with deeper correlation estimates would help distinguish how much of the degradation originates from the limited information available at $p=1$ and how much from its propagation through recursion. Increasing the QAOA depth may improve the correlation estimates by accessing a larger neighbourhood, although doing so would move away from the shallow-circuit regime that originally motivates RQAOA.

This degradation could be detected largely because of our choice of metric. The paint-swap ratio $\mathbb{E}[\Delta_C/n]$ is scale-normalised and requires no oracle access to the optimum, enabling consistent evaluation at arbitrarily large $n$. In contrast, the approximation ratio for MaxCut requires the exact optimum, which is infeasible at scale and can mask size-dependent pathologies. While one might object that these are distinct metrics, our formal reduction of BPSP to MaxCut shows that the observed size-dependent degradation occurs on a family of signed weighted MaxCut/Ising instances arising from BPSP. Whether analogous deterioration occurs more generally for MaxCut or Ising models remains an open question.

\subsection{Empirical Scaling of Level-1 XQAOA}

In contrast to the deteriorating performance of RQAOA$_1$, XQAOA$_1$
exhibits remarkably stable behaviour across the full range of problem
sizes considered. As shown in \cref{xqaoa_swap,fig:example_b}, for each
$n$ the distribution over the optimised solutions obtained from all
random initialisations remains centred near a paint-swap ratio of
$0.357$. Moreover, the distributions become increasingly concentrated
as $n$ grows, while their means remain essentially unchanged. Thus,
within the range of sizes accessible to our simulations, we observe no
evidence of size-dependent degradation in XQAOA$_1$; instead, its
optimised paint-swap ratios appear to concentrate around a
size-independent empirical value.

Importantly, the observed value of approximately $0.357$ is not an
expressivity limit of the XQAOA$_1$ ansatz. For the $X=Y$ mixer used
throughout this work, set all problem-unitary angles to zero,
$\gamma_{ij}=0$. The problem unitary then acts trivially, and the
initial state remains $|+\rangle^{\otimes n}$ before the local mixing
operations. On each qubit, choosing the shared mixer angle
$\alpha_i=\beta_i=\pi/4$ maps $|+\rangle_i$ to $|1\rangle_i$ up to a
global phase, whereas choosing
$\alpha_i=\beta_i=3\pi/4$ maps it to $|0\rangle_i$ up to a global
phase. Consequently, by choosing these angles independently on each
qubit, XQAOA$_1$ can prepare any computational-basis state. Since the
ICC encoding provides a bijection between computational-basis
assignments and valid BPSP colourings, the globally optimal
XQAOA$_1$ parameters can in principle represent an exact optimal
solution of every BPSP instance.

The value near $0.357$ should therefore be interpreted as a property
of the optimisation procedure rather than of the representational
capacity of the ansatz. In particular, it characterises the solutions
typically reached by the gradient-based optimisation from random
initialisations used in our benchmarks. The simultaneous stability of
the mean and narrowing of the distribution with increasing $n$ is
consistent with the optimised objective concentrating around a
characteristic value for this ensemble and optimisation procedure,
although our numerics alone do not establish an asymptotic limiting
constant.

A possible explanation for this robust optimisation behaviour is
provided by the increased expressiveness and favourable optimisation
landscapes identified for XQAOA in Ref.~\cite{vijendran2023expressive}.
There, overparameterisation was found to produce comparatively benign
loss landscapes in which poor local optima are suppressed and good
solutions are reached reliably from random initialisations. The
present BPSP results are consistent with this picture, but do not
establish that the same mechanism is responsible for the behaviour
observed here. In particular, the BPSP graphs have a distinctive
signed weight structure, and we have not analysed how the optimised
parameters correlate with these edge weights. Further analysis of the
optimised parameters and their relation to the signed edge weights may
help clarify whether this interpretation applies specifically to the
BPSP.

\subsection{Beyond the Benchmarked Algorithms}

In this work, we restricted our classical benchmarks to Red-First, Greedy, and Recursive Greedy, whose asymptotic performances are proven, and Recursive Star Greedy, whose asymptotic performance is conjectured. We additionally consider shallow-depth QAOA variants previously shown to outperform the best classical algorithms on specific MaxCut instance families. Other algorithms, while lacking proven performance for BPSP, could in principle achieve lower empirical paint-swap ratios. A natural candidate is the Goemans–Williamson (GW) algorithm~\cite{goemans1995improved}: since BPSP reduces to a weighted MaxCut instance, GW can be applied directly. However, two caveats arise. First, the presence of negative edge weights prevents the constant-factor approximation guarantee from carrying over. Second, GW requires solving a large semidefinite program whose memory demands grow rapidly with system size, often exceeding the capacity of standard hardware. For this reason we did not include GW in our large-scale comparisons. Nevertheless, on smaller instances it can perform marginally better than XQAOA$_1$. On our $n=128$ benchmark set, using 100 hyperplane roundings for each of 50 random BPSP instances, GW achieved an average paint-swap ratio of $0.354 \pm 0.035$, compared with $\approx 0.357$ for XQAOA$_1$. This small advantage, however, comes at the cost of solving a large semidefinite program.

Beyond this, there exist many untested heuristics that could achieve lower empirical ratios on specific instances, including simple meta-heuristics such as randomised local search, albeit without provable guarantees. Finally, we emphasise that the empirical XQAOA$_1$ paint-swap ratio of approximately $0.357$, observed consistently throughout the tested range, remains well above the recently established lower bound on the optimal solution, $E_{\mathrm{OPT}}[\Delta_C]\geq0.214n-o(n)$~\cite{hanvcl2023improved}. Since the lower bound is not known to be tight, this does not determine how close current algorithms are to optimality; nevertheless, it leaves open the possibility of further algorithmic improvements---classical, quantum, or hybrid.

\section{Conclusion} \label{conc_sec}

In this work, we have shown how the Binary Paint Shop Problem can be
written in a form suitable for Hamiltonian-based quantum optimisation.
Using the Initial Car Colour encoding, we reduce BPSP exactly to a signed
weighted MaxCut problem while halving the number of decision variables,
and from this obtain an Ising Hamiltonian whose ground states correspond
to optimal BPSP colourings. The reduction also makes several structural
features of the problem explicit. In particular, the resulting BPSP
graphs have degree at most four, their edge weights belong to
$\{\pm1,-2\}$, and their total weight can be related directly to the cost
of simple BPSP colourings. Together, these results connect the original
paint-shop problem to a structured family of signed Ising models that can
be studied using quantum optimisation algorithms.

Our large-scale benchmarks show a clear difference between XQAOA$_1$
and RQAOA$_1$ as the problem size increases. For XQAOA$_1$, the
paint-swap distributions obtained from all 100 optimised restarts remain
centred near $0.357$ and become more tightly concentrated as the instances
grow. In the instance-wise comparisons, XQAOA$_1$ outperforms
RQAOA$_1$ on almost all instances and each of the four classical BPSP
heuristics on every instance. RQAOA$_1$ behaves very differently. Even
though provably near-optimal QAOA$_1$ parameters are used at every
recursive step, its mean paint-swap ratio increases steadily from $0.322$
at $n=128$ to $0.390$ at $n=4096$. At the largest sizes it falls behind
Recursive Star Greedy and approaches the performance of Recursive Greedy.
We also find that Recursive Star Greedy has an empirical average near
$0.370$, which is above its conjectured asymptotic value of $0.361$ but
below the previously reported QAOA$_7$ ratio of $0.393$. The earlier
QAOA$_7$ result was therefore sufficient to outperform Recursive Greedy,
which was the strongest classical heuristic considered at the time, but it
does not outperform Recursive Star Greedy at the paint-swap ratios observed
here.

The observed XQAOA$_1$ paint-swap ratio of $0.357$ therefore reflects
the optimisation procedure used here rather than a limitation in the
representational capacity of the ansatz. With the $X=Y$ mixer,
XQAOA$_1$ can represent any computational-basis assignment and can
therefore, in principle, represent an optimal colouring for every BPSP
instance. The observed concentration instead reflects the solutions
typically reached by the gradient-based optimisation used in our
benchmarks. The stability of these solutions is consistent with the
greater expressiveness and favourable optimisation landscapes reported
in earlier studies of XQAOA, although we have not established that these
features are responsible for the behaviour seen here. For RQAOA$_1$,
one possible explanation for the deterioration is that the local
correlations available at $p=1$ can lead to imperfect elimination
decisions, which are then carried forward through the recursion and
allowed to accumulate. Our simulations do not separate these effects,
however, so the origin of the degradation remains open.

These results also suggest several natural questions for future work. For
RQAOA, one could test whether using deeper correlation estimates reduces
the observed deterioration, or compare against a matched non-recursive
scheme to separate the effect of shallow correlations from the effect of
recursively propagating the resulting decisions. For XQAOA, studying how
the optimised parameters depend on the signed edge weights of the BPSP
graphs may help explain why its performance remains so stable with problem
size. It would also be interesting to ask whether the same behaviour appears
in multi-colour paint-shop problems or in other structured signed Ising
models arising from industrial optimisation. More broadly, BPSP provides a
concrete family of instances on which two shallow quantum optimisation
methods behave very differently as the system size grows, making it a useful
setting for studying when recursive and overparameterised quantum heuristics
continue to perform well at scale.

\section*{Acknowledgements}

This research is supported by the Agency for Science, Technology and Research (A*STAR) under the Central Research Fund (CRF) Award for Use-Inspired Basic Research (UIBR), A*STAR C230917003, and the Quantum Innovation Centre (Q.InC) Strategic Research and Translational Thrust. This project is also supported by the National Research Foundation, Singapore through the National Quantum Office, hosted in A*STAR, under the Advanced Quantum Algorithms and Solutions (AQAS) Funding Initiative (S25Q9DA001). V.V. is supported by the National Research Foundation, Singapore, through the National Quantum Office, hosted in A*STAR, under its National Quantum Scholarship Scheme (NQSS) Funding Initiative (PhD). The authors thank the two anonymous reviewers of \textit{Physical Review Applied} for their careful reading of the manuscript and for constructive comments that improved the clarity and presentation of the results.

\section*{Code and Data Availability}

All code used to generate the numerical results, together with the
datasets, including the random BPSP instances used in our benchmarks,
are publicly available in Ref.~\cite{vijendran2026xqaoabpsp}.

\vfill

\clearpage

\newpage

\appendix
\onecolumngrid

\begin{center} \textbf{Appendices} \end{center}

In the following appendices, we delve into the mathematical ingredients required to formulate BPSP in a framework compatible with implementation using QAOA and its variants. Furthermore, we elucidate the relationship between BPSP and the weighted maximum cut (MaxCut) problem, examining key properties exhibited by the associated MaxCut graph.

\section{BPSP Formulation}
\label{sec:bpsp_formulation}

We begin by introducing some notation. Throughout the appendix, we fix $n \in \mathbb Z^+$ to be a positive integer that denotes the number of distinct symbols. Additionally, we let $r$ and $b$ be formal symbols representing the colours \textit{\red{red}} and \textit{\blue{blue}}, respectively.

\vspace{0.2cm}
\noindent\textbf{BPSP instances}: The instances of BPSP are \textit{double-occurrence} words, defined as even-length words over an alphabet $[n]= \{1,\ldots,n\}$ where each symbol $i\in [n]$ occurs exactly twice in each word. Each BPSP word of length $2n$ can be thought of as representing a sequence comprising $n$ cars, with each car appearing precisely twice within the sequence. Formally, we denote the set of such words of length $2n$ by
\begin{align}
    \Gamma_n = \left\{(x_1,x_2,\ldots,x_{2n}) \in [n]^{2n}: \forall i \in [n],
    |\{j\in [2n]: x_j=i\}|=2
    \right\}.
\label{eq:double_occurrence_words_n}
\end{align}
Note that the length of any word $x \in \Gamma_n$ is $|x|=2n$. The set of double-occurrence words, regardless of their length, is then the union of the sets in \eqref{eq:double_occurrence_words_n}:
\begin{align}
    \Gamma = \bigcup_{n \in \mathbb Z^+} \Gamma_n.
\end{align}

An example of an $n=3$ BPSP instance is
\begin{align}
\label{eq:BPSP_example}
    t = (1,2,1,3,3,2) \in \Gamma_3.
\end{align}

\noindent\textbf{BPSP colourings}:
A valid BPSP colouring of a BPSP instance $x \in \Gamma$ is one that colours the symbols of $x$ using two different colours in such a way that pairs of identical symbols are coloured differently. We let $\Xi$  denote the function that maps an instance $x \in \Gamma$ to the set of its valid colourings, i.e., 

\begin{align}
\label{eq:set_valid_colorings}
    \Xi (x) =
\left\{
f \in \{r,b\}^{|x|}: \forall i \neq j \in [|x|], x_i = x_j \implies f_i \neq f_j
\right\} .
\end{align}
We denote the restriction of $\Xi$ to the domain $\Gamma_n$ by $\Xi_n = \Xi|_{\Gamma_n}$, i.e.~if $x \in \Gamma_n$, we write $\Xi(x) = \Xi_n(x)$. For example, 
\begin{align}
    \label{eq:BPSP_example_coloring}
    rrbrbb, rbbbrr \in \Xi_3(t)
\end{align}
are valid BPSP colourings of the BPSP instance in \eqref{eq:BPSP_example}.

\vspace{0.2cm}
\noindent\textbf{BPSP cost function}:
The BPSP cost function $\xi$ measures the number of colour swaps in a given colouring of a BPSP instance. More formally, let $f\in \Xi(x)$ be a valid BPSP colouring of a BPSP instance $x \in \Gamma_n$. The BPSP cost $\xi$ of $f$ is given by
\begin{align}
\label{eq:BPSP_cost_function}
    \xi(f) = \sum_{i=1}^{2n-1} [f_i \neq f_{i+1}],
\end{align}
where $[\,\cdot\,]$ denotes the Iverson bracket which evaluates to 1 if its argument is true and 0 otherwise. When restricted to inputs of length $2n$, we write $\xi = \xi_n$, i.e. $\xi(f) = \xi_n(f)$ if $f\in\{r,b\}^{2n}$. For the valid BPSP colourings in our example \eqref{eq:BPSP_example_coloring}, we have
\begin{align}
\label{eq:BPSP_example_coloring_cost}
    \xi(rrbrbb) = 3, \quad \xi(rbbbrr) = 2. 
\end{align}
\vspace{0.2cm}
\noindent\textbf{BPSP problem}: With the above definitions, the BPSP problem \eqref{binary_paint_shop_problem} can formally be expressed as follows:
\newcounter{Problem}
\vspace{-0.2cm}
\par\noindent\rule{\textwidth}{0.5pt}
\textbf{\refstepcounter{Problem}Problem~\theProblem\label{problem:BPSP}.} \textsc{Binary Paint Shop Problem (BPSP)}
\hfill \nopagebreak\\[-6pt]\nopagebreak
\noindent\rule{\textwidth}{0.3pt} \nopagebreak\\[-10pt]
\begin{tabular}{rrl}
\textbf{Input:}     &  
\multicolumn{2}{l}{a BPSP instance $x \in \Gamma$.} \\
\textbf{Output:}
& minimise & \quad $\xi(f)$, as defined in \cref{eq:BPSP_cost_function}
\\
& subject to & \quad $f \in \Xi(x)$, where $\Xi(x)$ is defined in \cref{eq:set_valid_colorings}.
\end{tabular}
\\
\noindent\rule{\textwidth}{0.5pt}

For $x\in \Gamma$, we denote the optimal cost and the set of optimal solutions of Problem~\ref{problem:BPSP}, respectively, as follows:
\begin{align}
\label{eq:BPSP_function}
    \text{BPSP}(x) = \min_{f \in \Xi(x)} \xi(f), \\
    \label{eq:BPSP*}
    \text{BPSP}^*(x) = \argmin_{f \in \Xi(x)} \xi(f).
\end{align}
In the context of our example presented in \cref{eq:BPSP_example}, with some effort, one could demonstrate that $\BPSP(t) = 2$ and $\BPSP^*(t) \ni rbbbrr$.

\section{Fundamentals}

\subsection{Eta Function}

We begin this section by introducing the \textit{eta function}, a useful function that recurs throughout many subsequent expressions. The \textit{eta function} $\eta:\Gamma_n \times [2n-1] \to \{0,1\}$ is defined as follows:
\begin{align}
\label{eq:eta_function_def}
    \eta(x,i) = [x_{i+1}\in \{x_1,\ldots, x_i\}] \oplus [x_i \in \{x_1,\ldots, x_{i-1}\}],
\end{align}
where $\oplus$ is the \textit{exclusive or} operation defined by $a\oplus b = 0$ if $(a,b) = (0,0)$ or $(1,1)$, and $a\oplus b = 1$ otherwise. Equivalently,
\begin{align}
\eta(x, i) = 
\begin{cases} 
0, & \text{if } (x_{i+1} \in \{ x_1, \dots, x_i \} \land x_i \in \{ x_1, \dots, x_{i-1} \} )
\ \text{ or }\ (x_{i+1} \not\in \{ x_1, \dots, x_i \} \land x_i \not\in \{ x_1, \dots, x_{i-1} \}) \\
1, & \text{if } (x_{i+1} \in \{ x_1, \dots, x_i \} \land x_i \not\in \{ x_1, \dots, x_{i-1} \}) \ \text{ or }\ (x_{i+1} \not\in \{ x_1, \dots, x_i \} \land x_i \in \{ x_1, \dots, x_{i-1} \}).
\end{cases}
\end{align}

For our example word $t= (1,2,1,3,3,2)$ in \cref{eq:BPSP_example}, we have 
\begin{align}
    \eta(t,1) = \eta(t,5) = 0 \mbox{ and } \eta(t,2)=\eta(t,3)=\eta(t,4)=1.
\end{align}

We next establish a useful property of the eta function: if a BPSP instance denoted by $x\in\Gamma$ exhibits consecutive occurrences of the same symbol, with the initial occurrence positioned at index $i$, it follows that $\eta(x,i)=1$.

\begin{lemma}
\label{eq:eta_consecutive_symbols}
Let $x \in \Gamma_n$ and $i \in [2n-1]$. If $x_i = x_{i+1}$, then $\eta(x, i) = 1$.
\end{lemma}

\begin{proof}
If $x \in \Gamma_n$ and $x_i = x_{i+1}$, then $x_{i+1} \in \{x_1,\ldots, x_i\}$ and $x_i \not\in \{x_1,\ldots, x_{i-1}\}$ since the symbol $x_i$ occurs at most twice. Hence, $\eta(x,i)=1$.
\end{proof}

\subsection{Colouring Schemes}
\label{eq:coloring_schemes}

A \textit{colouring scheme} refers to an algorithm designed to assign colours to a BPSP word in accordance with a valid BPSP colouring. These schemes can be categorised as either deterministic or probabilistic. In this section, we focus on deterministic colouring schemes, which can be formally characterised by functions $\sigma: \Gamma_n \rightarrow\{r, b\}^{2 n}$ satisfying $\sigma(x) \in \Xi_n(x)$. To exemplify our definition, we discuss four deterministic colouring schemes:
\begin{enumerate}
    \item The \textit{red-first colouring scheme} is the function $\sigma_{\RF}: \Gamma_n \rightarrow\{r, b\}^{2 n}$, where $\forall x \in \Gamma_n$ and $i \in [2n]$,
\begin{align}
\sigma_{\RF}(x)_i=
\begin{cases}
    r, & x_i \notin\{x_1, \ldots, x_{i-1}\} \\
    b,  & \mbox{otherwise.}
\end{cases}
\end{align}
In words, this scheme colours the initial instance of each symbol red and the subsequent occurrence of each symbol blue.
    \item The \textit{greedy colouring scheme} is the function $\sigma_{\Gr} : \Gamma_n \to \{ r, b \}^{2n}$, where $\forall x \in \Gamma_n$, $g = \sigma_\Gr(x) \in \{ r, b \}^{2n}$ is defined recursively by 
\begin{align}
g_1 &=r \\
g_i &= \begin{cases}
g_{i-1}, & 
x_i \notin\left\{x_1, \ldots, x_{i-1}\right\}, \\
\neg g_{\min \left\{j: x_j=x_i\right\}}, & \mbox{otherwise},
\end{cases}
\end{align}
for $i\geq 2$. In essence, this scheme employs a greedy approach to colour symbols in $x$ from left to right, switching colours only when maintaining the current colour would violate the validity of the BPSP colouring.
 
\item The \textit{recursive greedy colouring scheme} is a map $\sigma_{\RG}:\Gamma_n\to\{r,b\}^{2n}$ that assigns to every BPSP instance $x$
a feasible colouring $\sigma_{\RG}(x)\in\Xi_n(x)$ of its $2n$ car occurrences. Write the input as $x^{(n)}=(x^{(n)}_1,\ldots,x^{(n)}_{2n})$, where the superscript records the number of distinct cars. The construction proceeds in two stages. \emph{Reduction} iteratively removes, from right to left, the last car in the current sequence together with its earlier partner, recording the partner’s index and thereby producing a chain of shorter instances $x^{(n)}\!\to x^{(n-1)}\!\to\cdots\to x^{(1)}$. \emph{Reconstruction} then rebuilds the colouring from length $2$ up to $2n$: at each step it copies all already fixed colours into the longer sequence and assigns colours to the two
reinserted cars using only local information. We now formalise the two stages, beginning with the reduction step.

\noindent\textbf{Reduction:} Starting from $x^{(n)}$, peel off pairs from right to left; at each step note where the partner of the last car sits so the pair can be reinserted in the correct gap later. Formally, for $j=n-1,n-2,\ldots,1$:
    \begin{enumerate}
        \item Let $m_{j+1}\in\{1,2,\ldots,2j+1\}$ be the position of the earlier occurrence of the last car of $x^{(j+1)}$, i.e.\
        $x^{(j+1)}_{m_{j+1}}=x^{(j+1)}_{2j+2}$.
        \item Delete both occurrences to obtain $x^{(j)}\in\Gamma_j$:
        \begin{equation*}
        x^{(j)}_k=
        \begin{cases}
        x^{(j+1)}_k,& k=1,2,\ldots,m_{j+1}-1,\\[2pt]
        x^{(j+1)}_{k+1},& k=m_{j+1},m_{j+1}+1,\ldots,2j.
        \end{cases}
        \end{equation*}
    \end{enumerate}

\noindent\textbf{Reconstruction:} Begin from the unique length-two colouring and, at each extension, reinsert the pair at the recorded gap, copy all previously fixed colours unchanged, and choose the two new colours so as not to introduce avoidable swaps. Concretely:
\begin{enumerate}
    \item \textit{Base Case.} Set $g^{(1)}=(r,b)\in\Xi_1(x^{(1)})$.
    \item \textit{Recursive Step.} For $j=2,3,\ldots,n$:
    \begin{enumerate}
        \item \emph{Embed the previously assigned colours} by copying the colouring from the shorter instance into $x^{(j)}$ at all positions except the two reinserted occurrences, which occupy indices $m_j$ and $2j$:
        \begin{equation*}
        g^{(j)}_k=
        \begin{cases}
        g^{(j-1)}_k,& k=1,2,\ldots,m_j-1,\\[2pt]
        g^{(j-1)}_{k-1},& k=m_j+1,m_j+2,\ldots,2j-1.
        \end{cases}
        \label{eq:rg_embed}
        \end{equation*}
        \item \emph{Colour the two reinserted cars}—the first occurrence placed at index $m_j$ and the second occurrence placed at index $2j$. Let $L:=g^{(j)}_{m_j-1}$ and $R:=g^{(j)}_{m_j+1}$ denote the colours of the cars immediately to the left and right of the gap (when they exist), and let $T:=g^{(j)}_{2j-1}$ be the colour of the last pre-existing car. Set
        \begin{align*}
        &\text{if }m_j=1: && g^{(j)}_{m_j}=R,\qquad g^{(j)}_{2j}= \neg R;\\
        &\text{if }m_j=2j-1: && g^{(j)}_{m_j}=L,\qquad g^{(j)}_{2j}= \neg L;\\
        &\text{if }1 < m_j < 2j-1 \text{ and } L=R:
        && g^{(j)}_{m_j}=L,\qquad g^{(j)}_{2j}= \neg L;\\
        &\text{if }1 < m_j < 2j-1 \text{ and } L\neq R:
        && g^{(j)}_{2j}=T,\qquad g^{(j)}_{m_j}= \neg T.
        \end{align*}
    \end{enumerate}
\end{enumerate}

\noindent
This recursion yields a unique colouring $g^{(n)}$, ensuring that the two occurrences of each car receive opposite colours (so $g^{(n)}\in\Xi_n(x^{(n)})$), and does not introduce additional colour changes beyond those that are unavoidable. We set
$\sigma_{\RG}(x^{(n)}) := g^{(n)}$.

\item The \textit{recursive star greedy} colouring is a map $\sigma_{\RSG}:\Gamma_n\to\{r,b\}^{2n}$ that returns, for every BPSP instance
$x$, a feasible colouring $\sigma_{\RSG}(x)\in\Xi_n(x)$. It augments the recursive greedy scheme by allowing a temporary symbol $*$ (``star”) during the construction to mark a car whose colour choice does not affect the current number of paint swaps. Stars are used only as placeholders: a finalisation step replaces each $*$ by $r$ or $b$ without changing the swap count and while preserving the opposite-colour constraint for the two occurrences of every car. Write the input as $x^{(n)}=(x^{(n)}_1,\ldots,x^{(n)}_{2n})$. The algorithm has two stages. \emph{Reduction} removes, from left to right, the first car of the current sequence together with its later partner, records the partner’s index, and thereby produces a chain of shorter instances $x^{(n)}\!\to x^{(n-1)}\!\to\cdots\to x^{(1)}$. \emph{Reconstruction} then rebuilds the colouring from length $2$ up to $2n$ by copying previously fixed colours and assigning colours—or, when locally indifferent, stars—to the previously reinserted cars using only local information. We now formalise these two stages, beginning with the reduction step.

\noindent\textbf{Reduction:} Starting from $x^{(n)}$, peel off pairs from left to right; at each step note where the partner of the first car sits so the pair can be reinserted in the correct gap later. Formally, for $j=n-1,n-2,\ldots,1$:
\begin{enumerate}
\item Let $m_j\in\{2,3,\ldots,2j+2\}$ be the position of the partner of the first car of $x^{(j+1)}$, i.e.\ $x^{(j+1)}_{1} = x^{(j+1)}_{m_{j}}$. 
\item Delete positions $1$ and $m_j$ to obtain $x^{(j)}\in\Gamma_j$:
    \begin{equation*}
        x^{(j)}_k=
            \begin{cases}
                x^{(j+1)}_{k+1},& k=1,2,\ldots,m_j-2,\\[2pt]
                x^{(j+1)}_{k+2},& k=m_j-1,m_j,\ldots,2j,
            \end{cases}
        \label{eq:rsg_shrink}
    \end{equation*}
and record $m_j$ (the future insertion position of the second occurrence).
\end{enumerate}

\noindent\textbf{Reconstruction:} Start from the unique length-two colouring; at each extension, reinsert the recorded pair, copy all previously fixed colours unchanged, and assign colours to the two new cars so as to avoid any unnecessary swaps. When a local tie occurs—i.e., both assignments give the same swap count—mark the indifferent choice with a star $*$. After the final step, replace every remaining $*$ arbitrarily by $r$ or $b$, choosing the replacements so that the two occurrences of each car have opposite colours and the swap count is unchanged. Concretely:
\begin{enumerate}
\item \textit{Base Case.} Set $h^{(1)}=(r,b)\in\Xi_1(x^{(1)})$.
\item \textit{Recursive Step.} For $j=2,3,\ldots,n$:
\begin{enumerate}
\item \emph{Embed the previously assigned symbols} by copying the colouring from the shorter instance into $x^{(j)}$ at all positions except the two reinserted occurrences, which occupy indices $1$ and $m_{j-1}$:
\begin{equation*}
h^{(j)}_k=
\begin{cases}
h^{(j-1)}_{k-1},& k=2,3,\ldots,m_{j-1}-1,\\[2pt]
h^{(j-1)}_{k-2},& k=m_{j-1}+1,\ldots,2j.
\end{cases}
\label{eq:rsg_embed}
\end{equation*}

\item \emph{Colour the two reinserted cars}—the first occurrence at index $1$ and the second at index $m_{j-1}$. Let $L:=h^{(j)}_{m_{j-1}-1}$ and $R:=h^{(j)}_{m_{j-1}+1}$ be the colours immediately to the left and right of the car at $m_{j-1}$ (when they exist), and let $N:=h^{(j)}_{2}$ be the colour of the second car. If exactly one of $L$ or $R$ is a star, write $s\in\{-1,+1\}$ for the side of that star (so $h^{(j)}_{m_{j-1}+s}=*$) and set $U:=h^{(j)}_{m_{j-1}-s}$ for its second (non-starred) neighbour. Apply the following local rules:
\begin{minipage}{\dimexpr\textwidth-2cm}
\begin{align*}
    &\text{if }m_{j-1}=2: && h^{(j)}_{m_{j-1}}=R,\qquad h^{(j)}_{1}= \neg R;\\
    &\text{if }m_{j-1}=2j: && h^{(j)}_{m_{j-1}}=b,\qquad h^{(j)}_{1}= r;\\
    &\text{if }2<m_{j-1}<2j \text{ and } L=R: && h^{(j)}_{m_{j-1}}=L,\qquad h^{(j)}_{1}= \neg L;\\
    &\text{if }2<m_{j-1}<2j \text{ and } L\neq R: && h^{(j)}_{1}=N,\qquad h^{(j)}_{m_{j-1}}= \neg N;\\
    &\text{if }2<m_{j-1}<2j \text{ and exactly one of } \{L,R\} \text{ is } * \text{ and } U\neq N: && h^{(j)}_{1}=N,\qquad h^{(j)}_{m_{j-1}}= \neg N;\\
    &\text{if }2<m_{j-1}<2j \text{ and exactly one of } \{L,R\} \text{ is } * \text{ and } U= N: && h^{(j)}_{1}=N,\qquad h^{(j)}_{m_{j-1}} = h^{(j)}_{m_{j-1}+s} = \neg N.
\end{align*}
\end{minipage}

\item \emph{Star Creation.} Let $b$ be the index of the second occurrence of the second car in the sequence $x^{(j)}$. Define the adjacent colours as follows: 
\begin{equation*}
L_1:=h_{1}^{(j)}, \quad R_1:=h_{3}^{(j)}, \quad L_2:=h_{b-1}^{(j)}, \quad R_2:=h_{b+1}^{(j)} .
\end{equation*}
Apply the following rules:
\begin{align*}
        &\text{if }L_1 = R_1 \text{ and } L_2 = R_2: && h^{(j)}_{2} = h^{(j)}_{b} = *;\\
        &\text{if }L_1 \neq R_1 \text{ and } L_2 \neq R_2: && h^{(j)}_{2} = h^{(j)}_{b} = *.
\end{align*}
\end{enumerate}
\item \textit{Breaking Ties.} After all pairs have been reinserted, replace each remaining $*$ arbitrarily with $r$ or $b$, subject to the opposite-colour constraint: the two occurrences of any car must receive different colours. This choice does not affect the total number of paint swaps.
\end{enumerate}

\noindent 
The recursion together with the star finalisation step determines a colouring $h^{(n)}$. It is feasible, so $h^{(n)}\in\Xi_n(x^{(n)})$, and by construction no additional paint swaps are created unless they are forced by the local neighbourhood. We set $\sigma_{\RSG}(x^{(n)}) := h^{(n)}$. The recursive star greedy method is a local strategy: resolving ties with stars and committing based only on current information can block better future choices. Consequently, there are instances where recursive star greedy performs worse than recursive greedy, and in some cases even worse than the basic greedy heuristic. 

\end{enumerate}
Note that in all colouring schemes, we consistently follow the convention of colouring the first symbol of $x$ in red.

\subsection{BPSP Costs of Red-First colouring} \label{redfirst_cost_app}

In this section, we focus on the red-first colouring scheme and show that the BPSP cost of this colouring can be expressed in terms of the eta function.

\begin{proposition}
    Let $x \in \Gamma_n$. Then the BPSP cost of the red-first colouring of $x$ is
    \begin{align}
    \label{eq:red-first-cost}
    \xi_n\left(\sigma_{\RF}(x)\right)=\sum_{i=1}^{2 n-1} \eta(x, i)
    = n - \frac{1}{2} - \frac{1}{2} \sum_{i=1}^{2n-1} (-1)^{\eta(x,i)} 
    .
    \end{align}
\end{proposition}
\begin{proof}
We start by proving the first equality. Let $f=\sigma_{\RF}(x)$. Then, we have that
\begin{align}
f_i & = \begin{cases}r, & x_i \notin\left\{x_1, \ldots, x_{i-1}\right\} \\
b, & x_i \in\left\{x_1, \ldots, x_{i-1}\right\}\end{cases} \\
\mbox{and }
f_{i+1} & = \begin{cases}r, & x_{i+1} \notin\left\{x_1, \ldots, x_i\right\} \\
b, & x_{i+1} \in\left\{x_1, \ldots, x_i\right\}.\end{cases}
\end{align}
Hence, it follows that:
\begin{align}
f_i \neq f_{i+1} & \iff (f_i = r \land f_{i+1} = b) \vee (f_i = b \land f_{i+1} = r) \nonumber \\
                 & \iff (x_i \notin \{x_1, \ldots, x_{i-1}\} \land x_{i+1} \in \{x_1, \ldots, x_i\}) \nonumber \vee (x_i \in \{x_1, \ldots, x_{i-1}\} \land x_{i+1} \notin \{x_1, \ldots, x_i\}) \nonumber \\
                 & \iff [x_{i+1} \in \{x_1, \ldots, x_i\}] \oplus [x_i \in \{x_1, \ldots, x_{i-1}\}] = 1 \nonumber \\
                 & \iff \eta(x, i) = 1 .
\end{align}
Therefore, the sum $\xi_n(f)$ can be calculated as:
\begin{align}
\xi_n(f) &= \sum_{i=1}^{2n-1} [f_i \neq f_{i+1}] \\
       &= \sum_{i=1}^{2n-1} [\eta(x, i) = 1] \\
       &= \sum_{i=1}^{2n-1} \eta(x, i), \quad \text{since } \eta(x, i) \in \{0,1\}.
\end{align}

To prove the second equality, we use the fact that if $z \in\{0,1\}$, then $ z = \dfrac{1 - (-1)^z}{2}$. This allows us to rewrite the above expression as follows:
\begin{align*}
\xi_n (f) &= \sum_{i=1}^{2n-1} \eta(x,i) = \sum_{i=1}^{2n-1} \frac{1 - (-1)^{\eta(x,i)}}{2} 
    = n - \frac{1}{2} - \frac{1}{2} \sum_{i=1}^{2n-1} (-1)^{\eta(x,i)}.
\end{align*}

\end{proof}

\section{Initial-Car-Colour (ICC) Encoding} \label{icc_encoding_scheme_app}

\subsection{ICC Encoding Scheme}
\label{sec:icc_encoding_scheme}

In \cref{eq:set_valid_colorings}, we defined a valid BPSP colouring of a BPSP instance $x \in \Gamma_n$ using a $2n$-bit string $f \in \{r,b\}^{2n}$. However, this representation carries excess information. Merely indicating the colouring of the initial occurrence of each symbol in $x$ is adequate for reconstructing the valid colouring. In other words, a succinct description of the colouring requires only $n$ bits. To formally express the relationship between these two representations, we introduce the function $\mathcal{E}_n \colon \Gamma_n \times \{r, b\}^n  \to \{r, b\}^{2n}$ defined as:
\begin{align}
\label{eq:En_function}
    \mathcal{E}_n(x, z)_i &=\neg^{\left[x_i \in\left\{x_1, \ldots, x_{i-1}\right\}\right]} z_{x_i} \\
    &= \begin{cases}
z_{x_i}, & \text{if } x_i \notin \{x_1, \ldots, x_{i-1}\}, \\
\bar{z}_{x_i}, & \text{if } x_i \in \{x_1, \ldots, x_{i-1}\},
\end{cases} \quad \text{for } i \in [2n].
\end{align}
In words, $\mathcal{E}_n$ takes as input a BPSP instance and the succinct $n$-bit description of a BPSP colouring and outputs the $2n$-bit description used in 
\eqref{eq:set_valid_colorings}. Also, we define $\mathcal{E} = \bigcup_{n \in \mathbb{N}} \mathcal{E}_n$ to be the union of the functions $\mathcal E_n$ (viewed as a binary relation), i.e., $\mathcal E(x,z) = \mathcal E_n(x,z)$ if $x \in \Gamma_n$ and $z\in \{r,b\}^n$. For convenience, we shall refer to the elements $f \in \Xi(x)$ defined in  \cref{eq:set_valid_colorings} as BPSP colourings in the \textit{standard encoding} and refer to the succinct colouring $\mathcal E(x,z)$ as BPSP colourings in the \textit{ICC encoding}.

Now, note that if we fix the first argument of $\mathcal E_n$, we get a function $\{r,b\}^n \rightarrow \{r,b\}^{2n}$ that maps the succinct $n$-bit description to the original $2n$-bit description. Formally, this may be described by the \textit{currying} function, which gives:
\begin{align}
\label{eq:curry_E_n}
    \mathrm{curry}(\mathcal{E}_n) \colon \Gamma_n \to \left(\{r, b\}^n \to \{r, b\}^{2n}
    \right),
\end{align}
where
\begin{align}
    \mathrm{curry}(\mathcal{E}_n)(x)(z) = \mathcal{E}_n(x, z).
\end{align}
Note that for $x\in \Gamma_n$, the function $\mathrm{curry}(\mathcal{E}_n)(x):\{r, b\}^{n}\rightarrow \{r, b\}^{2n}$ is not surjective, since not all $2n$-bit colourings are valid BPSP colourings. For simplicity, we will henceforth restrict the codomain of $\mathrm{curry}(\mathcal{E}_n)(x)$ to its range. Our next proposition establishes this notion and confirms that the range of $\curry(\mathcal E_n)(x)$ is $\Xi_n(x)$.
\begin{proposition} \label{icc_bijective_prop}
Let \( x \in \Gamma_n \). Then, the function
\begin{equation}
    \mathrm{curry}(\mathcal{E}_n)(x) \colon \{r, b\}^n \rightarrow \Xi_n(x)
\end{equation}
is invertible, with the inverse given by
\begin{align}
    \mathrm{curry}(\mathcal{E}_n)(x)^{-1} \colon \Xi_n(x) \rightarrow \{r, b\}^n,
\end{align}
where
\begin{equation}
    \mathrm{curry}(\mathcal{E}_n)(x)^{-1}(f)_t = f_{\min\{k \mid x_k = t\}}, \qquad \text{for } t \in [n].
\end{equation}
In other words, for all \( x \in \Gamma_n \), we have the following:
\begin{enumerate}
    \item[(i)] For all \( z \in \{r, b\}^n \), \( \mathcal{E}_n(x, z) = \mathrm{curry}(\mathcal{E}_n)(x)(z) \in \Xi_n(x) \),
    \item[(ii)] Let \( f \in \Xi_n(x) \) and \( z = \left(f_{\min\{k \mid x_k = t\}}\right)_{t=1}^n \). Then
    \begin{align}
        \mathcal{E}_n(x, z) := \mathrm{curry}(\mathcal{E}_n)(x)(z) = f,
    \end{align}
\end{enumerate}
That is, one could write
    \begin{align}
        \Xi_n(x) = \{ \mathcal{E}_n(x, z)  \mid z \in \{r, b\}^n \}.
    \end{align}
\end{proposition}

\begin{proof} 
We prove the two statements (i) and (ii) of the proposition separately:

\begin{enumerate}
    \item[(i)] 
Let $x \in \Gamma_n$ and  $z \in\{r, b\}^n$. Choose distinct indices $i<j \in[2 n]$ such that $x_i=x_j$. In order to demonstrate that $\mathcal{E}_n(x, z) \in \Xi_n(x)$, we employ a proof by contradiction. Suppose that $\mathcal{E}_n(x, z)_i = \mathcal{E}_n(x, z)_j$. Then, the following implications hold
\begin{align*}
    & \neg^{\left[x_i \in\left\{x_1, \ldots, x_{i-1}\right\}\right]} z_{x_i}=\neg^{\left[x_j \in\left\{x_1, \ldots, x_{j-1}\right\}\right]} z_{x_j} \\
    &\implies \neg^{\left[t \in\left\{x_1, \ldots, x_{i-1}\right\}\right]} z_t=\neg^{\left[t \in\left\{x_1, \ldots, x_{j-1}\right\}\right]} z_t \\
    &\implies [t \in \{x_1,\ldots,x_{i-1}\}] = [t \in \{x_1,\ldots,x_{j-1}\}] \\
    &\implies \underbrace{t \in \{x_1,\ldots,x_{i-1}\} \land t \in \{x_1,\ldots,x_{j-1}\}}_{\circled{ \scriptsize 1}}  \quad \text{or} \quad \underbrace{t \notin \{x_1,\ldots,x_{i-1}\} \land t \notin \{x_1,\ldots,x_{j-1}\}}_{\circled{ \scriptsize 2}}.
\end{align*}

Now, case \circled{ \scriptsize 1} implies that there exists 
$k \in\{1, \ldots, i-1\}$ such that $t=x_k$. But this would then mean that $x_k = x_i = x_j = t$ for distinct $k<i<j$, which then implies that $\left|\left\{j: x_j=t\right\}\right| \geqslant 3$, which contradicts our assumption that $x \in \Gamma_n$. On the other hand, case \circled{ \scriptsize 2} implies that $t \not\in \{x_1,\ldots,x_{j-1}\}$. But since $i<j$, we have $\{x_1,\ldots,x_{j-1} \}\ni x_i = t$, which contradicts \circled{ \scriptsize 2}.

Since both cases lead to a contradiction, $\mathcal E_n(x,z)_i \neq \mathcal E_n(x,z)_j $, which contradicts our supposition. Therefore, $\mathcal{E}_n(x, z) \in \Xi_n(x)$.

\item[(ii)] Let $x \in \Gamma_n$ represent a BPSP instance, and let $f \in \Xi_n(x)$ denote a valid BPSP colouring. For $t \in [n]$, define $z_t=f_{\min \left\{k: x_k=t\right\}}$.

We aim to show that 
    $\mathcal E_n(x, z)=f$. To establish this, we note that for $i \in [2n]$,
\begin{align}
    \mathcal{E}_n(x, z)_i &= \neg^{\left[x_i \in\left\{x_1, \ldots, x_{i-1}\right\}\right]} z_{x_i}  \\
    &= \neg^{\left[x_i \in\left\{x_1, \ldots, x_{i-1}\right\}\right]} f_{\min \left\{k: x_k=x_i\right\}}. \label{eqn_star}
\end{align}
Since $x \in \Gamma_n$ we have that $|\{k \colon x_k = x_i\}| = 2$. Hence, let $l$ be the unique element in $\{k \colon x_k = x_i\}$ that is not $i$. In other words, $\{k \colon x_k = x_i\} = \{i, l\}$, where $l \neq i$. We now consider the following two cases.

\vspace{0.2cm}
\noindent\underline{Case 1: $i < l$}. In this case, 
we have $x_i \notin \{x_1, \ldots, x_{i-1}\}$ (otherwise $\exists k < i \colon x_k = x_i$), and \( \min\{k \colon x_k = x_i\} = i \).

Hence, by \cref{eqn_star} it follows that
\begin{align}
    \mathcal{E}_n(x,z)_i 
    = \neg^0 f_i = f_i.
\end{align}

\underline{Case 2: \( i > l \)}: In this case, \( x_i \in \{x_1, \ldots, x_l, \ldots, x_{i-1}\} \) since \( l \in \{1, \ldots, i-1\} \). Furthermore, \( \min\{k \colon x_k = x_i\} = l \). Hence, by \cref{eqn_star} we have that
\begin{align}
    \mathcal{E}_n(x, z)_i 
    = \neg f_l
    = f_i.
\end{align}

Since $\mathcal{E}_n(x,z)_i = f_i$ for all $ i \in [2n]$, it follows that $\mathcal{E}_n(x, z) = f$.
\end{enumerate}
This completes our proof of the proposition.
\end{proof}

\subsection{BPSP Cost in ICC Encoding} \label{sec:bpsp_cost_icc}

In \cref{eq:BPSP_cost_function}, we expressed the BPSP cost function $\xi$ as a function of the BPSP colouring $f \in \Xi(x)$. With the introduction of the ICC encoding scheme in \cref{sec:icc_encoding_scheme}, our objective in this section is to express the BPSP cost function directly in terms of the ICC-encoded colouring $\mathcal E(x,z)$. To this end, we define
the function $\tilde{\xi}_n : \Gamma_n \times \{r, b\}^n   \rightarrow \mathbb{N}$ as \begin{align}
\label{eq:xi_tilde}
    \tilde{\xi}_n = \xi_n \circ \mathcal E_n
\end{align}
and $\tilde \xi = \bigcup_{n\in \mathbb N} \tilde{\xi}_n$. Expressed in words, $\tilde \xi_n(x,z)$ denotes the tally of colour swaps when the BPSP instance $x$ is coloured using the colouring specified by $z$. The proof of the following proposition derives an explicit expression for $\tilde \xi_n(x,z)$.

\begin{proposition} \label{icc_cost_prop}
    Let $z \in \{r,b\}^n$ and $x \in \Gamma_n$. Then,
    \begin{align}
    \label{eq:tildexin}
        \tilde{\xi}_n (x, z) = \sum_{i=1}^{2n-1} \left[ \neg^{\eta(x,i)} z_{x_i} \neq z_{x_{i+1}} \right].
    \end{align}
\end{proposition}
\begin{proof}
Through a straightforward calculation, we determine:
    \begin{align*}
\tilde{\xi}_n (x, z) &= \xi_n (\mathcal{E}_n(x,z)) \\
&= \sum_{i=1}^{2n-1} \left[ \mathcal{E}_n(x,z)_i \neq \mathcal{E}_n(x,z)_{i+1} \right] \\
&= \sum_{i=1}^{2n-1} \left[ \neg^{[x_i \in \{x_1, \ldots, x_{i-1}\}]} z_{x_i} \neq \neg^{[x_{i+1} \in \{x_1, \ldots, x_{i}\}]} z_{x_{i+1}} \right]
\\
&= \sum_{i=1}^{2n-1} \left[ \neg^{[x_{i+1} \in \{x_1, \ldots, x_{i}\}] \oplus [x_i \in \{x_1, \ldots, x_{i-1}\}]} z_{x_i} \neq z_{x_{i+1}} \right] \\
&= \sum_{i=1}^{2n-1} \left[ \neg^{\eta(x,i)} z_{x_i} \neq z_{x_{i+1}} \right],
\end{align*}
where the first line is derived from 
\cref{eq:xi_tilde}, the second line from \cref{eq:BPSP_cost_function}, the third line from \cref{eq:En_function}, and the fifth line from the definition of the eta function as specified in \cref{{eq:eta_function_def}}. The fourth line follows from the properties $\neg^a \neg^b = \neg^{a\oplus b}$ and $\neg^0 = I$. This completes the proof of the proposition.
\end{proof}

\subsection{BPSP in ICC Encoding} \label{sec:bpsp_icc_encoding}

We now formulate the BPSP problem, as defined in Problem~\ref{problem:BPSP}, using the ICC encoding. We begin by observing the equivalence
$\{ \tilde{\xi}_n(x,z) : z \in \{r, b\}^n \} 
= \{ \xi_n (f) : f \in \Xi_n(x) \}$,
which implies that for any BPSP instance $x \in \Gamma_n$,
    \begin{align}
        \min_{f \in \Xi_n(x)} \xi_n(f) = \min_{z \in \{r, b\}^n} \tilde{\xi}_n (x, z).
    \end{align}

Hence, in parallel to the formulation presented in Problem~\ref{problem:BPSP}, the \textit{ICC formulation} of the BPSP problem can be expressed as
\begin{align}
\begin{array}{rll}
\text{Input:} & x \in \Gamma \\[0.2cm]
\text{Output:} &
\text{Minimize}  & \widetilde{\xi}(x, z) \\ &
\text{subject to} & z \in\{r, b\}^{|x| / 2}.
\end{array}
\end{align}
Also, \cref{eq:BPSP_function} can be written as
\begin{align}
\textsc{BPSP}(x)=\min _{f \in \Xi(x)} \xi(f)=\min _{z \in \{r, b\}^{|x| / 2}} \widetilde{\xi}(x, z).
\end{align}
Similarly to \cref{eq:BPSP*}, we define $\widetilde{\textsc{BPSP}^*}(x)$ as 
\begin{align}
\widetilde{\textsc{BPSP}^*}(x) := \underset{z \in\left\{r, b\right\}^{|x| / 2}}\argmin \tilde{\xi}(x, z) \quad \subseteq\{r, b\}^{|x| / 2}.
\end{align}

\subsection{ICC Encoding of Colouring Schemes} \label{sec:icc_encoding_colouring}

Our next objective is to reformulate the results of \cref{eq:coloring_schemes} for colouring schemes using the ICC encoding. This can be achieved as follows: Let $\sigma: \Gamma_n \rightarrow\{r, b\}^{2 n}$ be a BPSP colouring scheme. We define the \textit{ICC encoding} of $\sigma$ as $\operatorname{ICC}(\sigma): \Gamma_n \rightarrow\{r, b\}^n$, where
\begin{align}
\operatorname{ICC}(\sigma)(x)=\left(\operatorname{curry}(\mathcal{E}_n)(x)\right)^{-1}(\sigma(x)), \qquad \text { for } x \in \Gamma_n.
\end{align}
Equivalently,
\begin{align}
    \sigma(x)=\operatorname{curry}(\mathcal{E}_n)(x)\left(\operatorname{ICC}(\sigma)(x)\right)=\mathcal{E}_n\left(x,\operatorname{ICC}(\sigma )(x)\right).
\end{align}
Therefore,
\begin{align}
\xi(\sigma(x)) = \widetilde{\xi}(x,\operatorname{ICC}(\sigma)(x)).
\end{align}

Our next proposition demonstrates that the ICC encoding of the red-first colouring scheme maps any length-$2n$ BPSP instance to $r^n$, aligning with the prescribed rule of colouring the initial occurrence of each symbol in red.

\begin{proposition}
    Let $x \in \Gamma_n$. The ICC encoding of the red-first colouring $\sigma_{\mathrm{RF}}(x)$ of $x$ is
    \begin{align}
    \operatorname{ICC}\left(\sigma_{\operatorname{RF}}\right)(x)=r^n.
    \end{align}
\end{proposition}
\begin{proof}
    For $t \in [n]$, we have
\begin{align}
\operatorname{ICC}(\sigma_{\operatorname{RF}})(x)_t & =\operatorname{curry}(\mathcal E_n)(x)^{-1}\left(\sigma_{\mathrm{RF}}(x)\right)_t \nonumber\\
& =\sigma_{\mathrm{RF}}(x)_{\min \left\{k: x_k=t\right\}}.
\end{align}
Now, let $i = \min\{ k : x_k = t \}$. Then, $x_j = x_i \Rightarrow j > i$, which implies that
\begin{align*}
    & x_i \notin \{x_1, \ldots, x_{i-1}\}\\
    \implies \quad &  \sigma_{\mathrm{RF}}(x)_i = r \\
    \implies \quad & \operatorname{ICC}(\sigma_{\mathrm{RF}})(x)_t = r \quad \forall t \in [n] \\
    \implies \quad &
    \operatorname{ICC}\left(\sigma_{\operatorname{RF}}\right)(x)=r^n,
\end{align*}
completing the proof of the proposition.
\end{proof}

\section{BPSP Graph}
\label{sec:BPSP_graph}

In this appendix, we will demonstrate how every BPSP instance $x$ can be associated with a graph $G_x$, referred to as the BPSP graph. Through this association, we will establish that computing $\BPSP(x)$ from $x$ can be efficiently reduced to finding the maximum cut of $G_x$.

\subsection{BPSP Graph: Definition}

We begin by establishing some notation. For the BPSP instance $x=\left(x_1, x_2, \ldots, x_{2 n}\right) \in \Gamma_n$, we define the function $
\theta_x: 2^{[n]}  \to \mathbb{R}$, which maps subsets of $[n]$ to real numbers as follows:
\begin{align}
\label{eq:theta_x}
\theta_x(e)  =-\sum_{i=1}^{2 n-1}(-1)^{\eta(x, i)} \delta_{e,\left\{x_i, x_{i+1}\right\}},
\end{align}
where $2^{[n]}$ denotes the power set of $[n]$ and $\delta$ denotes the Kronecker delta. Now, we define the BPSP graph.

\renewcommand{\labelitemi}{\scriptsize{$\bullet$}}
\begin{definition}[BPSP graph]
\label{def:bpsp_graph}
Let $x=\left(x_1, x_2, \ldots, x_{2 n}\right) \in \Gamma_n \subseteq[n]^{2 n}$ be a BPSP instance. The \textit{BPSP graph} of $x$ is the weighted graph $G_x=\left(V_x, E_x, W_x\right)$, where
\begin{itemize}
    \item $V_x=[n]$ are the vertices,
    \item $E_x=\big\{\left\{x_{i}, x_{i+1} \right\}: i \in[2 n-1],\  x_i \neq x_{i+1},\ \theta_x\left(\left\{x_{i}, x_{i+1}\right\}\right) \neq 0\big\}$ are the edges,
    \item $W_x=\left.\theta_x\right|_{E_x}: E_x \to \mathbb R$ is the weight function assigning edges to real numbers.
\end{itemize}
\label{def:BPSP_graph}
\end{definition}

In essence, the BPSP graph associated with a length-$2n$ BPSP instance $x$ is constructed with vertices representing the symbols of the alphabet $[n]$. In this graph, edges connect symbols that are adjacent in $x$, with the exception of consecutive identical symbols and the zeroes of the $\theta_x$ function. According to Definition~\ref{def:BPSP_graph} and \cref{eq:theta_x}, the weight of an edge $e \in E_x \subseteq 2^{[n]}$ in the graph is given by
\begin{align}
W_x(e)=-\sum_{i=1}^{2 n-1}(-1)^{\eta\left(x,i\right)} \delta_{e,\left\{x_i, x_{i+1}\right\}}.
\end{align}

\subsection{Properties of BPSP Graph}

\subsubsection{Possible degrees of BPSP graph}

We begin by establishing some properties of BPSP graphs. To initiate, we show that the vertex degrees of a BPSP graph cannot exceed 4. Denoting the degree of a vertex $v$ in $G$ by
\begin{align}
    \operatorname{deg}_G(v) = \left| \{ u \in V: \{u,v\} \in E\} \right|,
\end{align}
we prove the following proposition.

\begin{proposition} \label{max_deg_prop}
    Let $x \in \Gamma_n$ be a BPSP instance. Then, for all vertices $v \in V_x$ in the BPSP graph of $x$, 
    \begin{align}
        \operatorname{deg}_{G_x}(v) \leq 4.
    \end{align} 
\end{proposition}
\begin{proof}
    Consider $v \in V_x = [n]$. Given that $x \in \Gamma_n$, it follows that the set of indices $
    \left\{j \in[2 n]: x_j=v\right\}
    $ contains precisely two elements, which we denote by $\alpha$ and $\beta$. By construction, $\alpha \neq \beta$, $x_\alpha=x_\beta=v$, and for any $\gamma$ such that $x_\gamma=v$, we have $\gamma \in\{\alpha, \beta\}$.

    Let $u\in V_x$ satisfy $\{u,v\}\in E_x$, i.e.\ $u$ and $v$ are
neighbours in the BPSP graph. Define the set of positions adjacent to the
two occurrences of $v$ by
\begin{equation}
    I_{\alpha,\beta}
    :=
    \{\alpha-1,\alpha+1,\beta-1,\beta+1\}\cap[2n].
\end{equation}
Since $x_\alpha=x_\beta=v$, the existence of the edge $\{u,v\}$ implies
that $u$ occurs next to one of these two occurrences in the BPSP word.
Hence,
\begin{equation}
    \{u,x_\alpha\}\in E_x
    \quad\text{or}\quad
    \{u,x_\beta\}\in E_x
    \quad\Longrightarrow\quad
    u=x_j
    \text{ for some }j\in I_{\alpha,\beta}
    \quad\Longrightarrow\quad
    u\in\{x_j:j\in I_{\alpha,\beta}\}.
    \label{eq:degree-neighbours}
\end{equation}
Therefore,
\begin{equation}
    \{u\in V_x:\{u,v\}\in E_x\}
    \subseteq
    \{x_j:j\in I_{\alpha,\beta}\},
\end{equation}
and consequently
\begin{align}
    \deg_{G_x}(v)
    &=
    \bigl|\{u\in V_x:\{u,v\}\in E_x\}\bigr| \nonumber \\
    &\le
    \bigl|\{x_j:j\in I_{\alpha,\beta}\}\bigr|\nonumber  \\
    &\le
    |I_{\alpha,\beta}| \nonumber  \\
    &\le 4,
\end{align}
which completes the proof of the proposition.
\end{proof}

\subsubsection{Total weight of BPSP graph}

Next, for any weighted graph $G=(V,E,W)$ (where $V$ is the set of vertices, $E$ is the set of edges, and $W:E \to \mathbb R$ is a weight function), we define the total weight of $G$ as the sum of the weights of all its edges:
\begin{align}
\label{eq:total_weight_def}
    \Wtot(G) =
    \sum_{\{i,j\}\in E} W\left(\{i,j\}\right).
\end{align}

Additionally, we introduce the concept of the \textit{double letter count} of a BPSP instance $x = (x_1, x_2, \ldots, x_{2n}) \in \Gamma_n \subseteq [n]^{2n}$, which is defined as the number of consecutive identical symbols in $x$, denoted as:
\begin{align}
\label{eq:dl}
    \dl(x) = \left|
        \{i \in [2n-1]: x_i = x_{i+1}
        \}
    \right|.
\end{align}

Our next proposition expresses the total weight of a BPSP graph $G_x$ using \cref{eq:dl} and the eta function defined in \cref{eq:eta_function_def}.

\begin{proposition}
    Given a BPSP instance $x \in \Gamma_n$, the total weight of the BPSP graph $G_x=\left(V_x, E_x, W_x\right)$ associated with $x$ is given by 
    \begin{align}
    \Wtot (G_x) = -\sum_{\substack{k=1 \\ x_k \neq x_{k+1}}}^{2 n-1}(-1)^{\eta(x, k)} =
    -\dl(x)
    -\sum_{k=1}^{2 n-1}(-1)^{\eta(x, k)}.
    \label{eq:total_weight_BPSP}
    \end{align}
    \label{prop:total_weight_BPSP}
\end{proposition}
\begin{proof} To establish the first equality presented in \cref{eq:total_weight_BPSP}, we proceed with the following calculation:

\begingroup
\allowdisplaybreaks
\begin{align}
    W_{\text {tot}}\left(G_x\right) 
    &=\sum_{\{a, b\} \in E_x} W_x(\{a, b\}) \nonumber\\
    &=\sum_{\{a, b\} \in E_x} \theta_x(\{a, b\}) \nonumber\\
    &=\sum_{\substack{
    {i \in [2 n-1]} \\ x_i \neq x_{i+1} \\ \theta_x\left(\left\{x_i, x_{i+1}\right\}\right) \neq 0}} \frac{\theta_x\left(\left\{x_i, x_{i+1}\right\}\right)}{\mid\left\{j \in[2 n-1]:\left\{x_j, x_{j+1}\right\}=\left\{x_i, x_{i+1}\right\} \mid\right.} \nonumber\\
    &=\sum_{\substack{i \in[2 n-1] \\ x_i \neq x_{i+1}}} \frac{\theta_x\left(\left\{x_i, x_{i+1}\right\}\right)}{\mid\left\{j \in[2 n-1]:\left\{x_j, x_{j+1}\right\}=\left\{x_i, x_{i+1}\right\} \mid\right.} \nonumber\\
    &=\sum_{\substack{i \in [2 n-1] \\ x_i \neq x_{i+1}}} \frac{
    -\sum_{k=1}^{2 n-1}(-1)^{\eta\left(x, k\right)} \delta_{\left\{x_i, x_{i+1}\right\},\left\{x_k, x_{k+1}\right\}}
    }{\mid\left\{{j \in [2 n-1}]:\left\{x_j, x_{j+1}\right\}=\left\{x_i, x_{i+1}\right\} \mid\right.} \nonumber\\
    &=-\sum_{k=1}^{2 n-1}(-1)^{\eta(x, k)} \sum_{\substack{i \in[2 n-1] \\ x_i \neq x_{i+1}}} \frac{\delta_{\left\{x_{i}, x_{i+1}\right\},\left\{x_k, x_{k+1}\right\}}}{\mid\left\{j \in[2 n-1]:\left\{x_j, x_{j+1}\right\}=\left\{x_i, x_{i+1}\right\} \mid\right.} \nonumber\\
    &=-\sum_{k=1}^{2 n-1}(-1)^{\eta(x, k)} \sum_{i \in[2 n-1]} [x_i \neq x_{i+1}] \frac{\delta_{\left\{x_i, x_{i+1}\right\},\left\{x_k, x_{k+1}\right\}}}{\mid\left\{j \in[2 n-1]:\left\{x_j, x_{j+1}\right\}=\left\{x_i, x_{i+1}\right\} \mid\right.} \nonumber\\
    &=-\sum_{k=1}^{2 n-1}(-1)^{\eta(x, k)} \frac{[x_k \neq x_{k+1}]}{\mid\left\{j \in[2 n-1]:\left\{x_j, x_{j+1}\right\}=\left\{x_k, x_{k+1}\right\} \mid\right.} \sum_{i \in[2 n-1]} \delta_{\left\{x_{i}, x_{i+1}\right\},\left\{x_k, x_{k+1}\right\}} \nonumber\\
    &=-\sum_{k=1}^{2 n-1}(-1)^{\eta(x, k)} \frac{
    [x_k \neq x_{k+1}]}{\mid\left\{j \in[2 n-1]:\left\{x_j, x_{j+1}\right\}=\left\{x_k, x_{k+1}\right\} \mid\right.} \left| \left\{i \in[2 n-1]:\left\{x_{i}, x_{i+1}\right\}=\left\{x_k, x_{k+1}\right\} \right| \right. \nonumber\\
    &=-\sum_{\substack{k=1 \\ x_k \neq x_{k+1}}}^{2 n-1}(-1)^{\eta(x, k)}.
    \label{eq:cal_total_weight}
\end{align}
\endgroup
The first line above follows directly from the definition of total weight given in \cref{eq:total_weight_def}. The second line stems from the relationship $W_x = \theta_x|_{E_x}$, as defined in Definition~\ref{def:bpsp_graph}. The third line follows from the observation that if a finite set $T = \{a_i: i \in M\}$ (where $M$ is also a finite set) contains elements $a_i$, with the possibility of some elements being identical, then for any function $f:T\to\mathbb R$:
\begin{align}
    \sum_{t\in T} f(t) = \sum_{i\in M} \frac{f(a_i)}{|\{j:a_i = a_j\}|}.
\end{align}
This indicates that when summing over a set with potentially repeated elements, each distinct element contributes its value divided by the number of times it appears in the set. The fourth line arises from the observation that the summand in the expression vanishes whenever $\theta_x(\{x_i,x_{i+1}\})=0$. The fifth line follows from the definition of $\theta_x$ provided in \cref{eq:theta_x}.
The sixth line is attained by rearranging the terms presented on the fifth line. The seventh line ensues from the property that if $S \subseteq T$ are any two indexing sets, then the following sums are equivalent:
\begin{align}
\label{eq:sum_Iverson}
    \sum_{i \in S} (\ \cdot\ ) = \sum_{i\in T} [i \in S] (\ \cdot \ ).
\end{align}
This enables us to represent the sum over the potentially smaller set $S$ as the sum over its superset $T$ by incorporating an indicator function. The eighth line follows from the observation that the sum in the expression vanishes whenever the sets $\{x_i,x_{i+1}\}$ and $\{x_k,x_{k+1}\}$ are not equal. The ninth line stems from the property that
\begin{align*}
    \sum_i \delta_{a_i a_j} f(a_j) = f(a_j)
    |\{i:a_i=a_j\}|.
\end{align*}
Finally, the tenth line is derived by applying \cref{eq:sum_Iverson} once more.
This concludes the proof of the first equality.

To show the second equality of \cref{eq:total_weight_BPSP}, consider:

\begingroup
\allowdisplaybreaks
\begin{align}
    \sum_{\substack{k=1 \\ x_k \neq x_{k+1}}}^{2 n-1}(-1)^{\eta(x, k)} &=\left(\sum_{k=1}^{2 n-1}-\sum_{\substack{k=1 \\ x_k=x_{k+1}}}^{2 n-1}\right)(-1)^{\eta(x, k)} \nonumber\\
    &= \sum_{k=1}^{2 n-1}(-1)^{\eta(x, k)}-\sum_{\substack{k=1 \\ x_k=x_{k+1}}}^{2 n-1}(-1)^{\eta(x,k)} \nonumber\\
    &= \sum_{k=1}^{2 n-1}(-1)^{\eta(x, k)}+\left|\left\{i \in[2 n-1]: x_i=x_{i+1}\right\}\right| \nonumber\\
    &= \sum_{k=1}^{2 n-1}(-1)^{\eta(x, k)}+\dl(x).
    \label{eq:second_equality}
\end{align}
\endgroup
Negating the equalities in \cref{eq:second_equality} yields the second equality in
\cref{eq:total_weight_BPSP}.

The first line above arises from the observation that the indices $k$ satisfying $x_k = x_{k+1}$ and those satisfying $x_k \neq x_{k+1}$ form a partition of the index set $[2n-1]$. The second line follows from the distributive property. The third line is deduced from the observation that the second summation is exclusively over terms where $x_k = x_{k+1}$, thereby implying $\eta(x,k)= 1$, as shown by Lemma~\ref{eq:eta_consecutive_symbols}. Finally, the fourth line follows from the definition of the double letter count as defined by \cref{eq:dl}.
\end{proof}

\subsubsection{Relating total weight of BPSP graph to BPSP cost of red-first colouring}

Our next corollary expresses the BPSP cost of the red-first colouring, calculated in \cref{eq:red-first-cost}, in terms of the double letter count as defined in \cref{eq:dl} and the total weight of the BPSP graph defined in \cref{eq:total_weight_def}.

\begin{corollary} \label{cor:red_first_cost_total_weight}
Let $x \in \Gamma_n$. Then the BPSP cost of the red-first colouring of $x$ is
\begin{align}
\label{eq:red_first_cost_total_weight}
\xi_n\left(\sigma_{\RF}(x)\right)=n-\frac{1}{2}+\frac{1}{2} \dl(x)+\frac{1}{2} W_{\mathrm{tot}}\left(G_x\right).
\end{align}
\end{corollary}

\begin{proof}
From \cref{eq:red-first-cost}, we derive the BPSP cost of the red-first colouring of $x$ as follows:
    \begin{align*}
        \xi_n\left(\sigma_{\RF}(x)\right) =n-\frac{1}{2}-\frac{1}{2} \underbrace{\sum_{i=1}^{2 n-1}(-1)^{\eta(x, i)}}_{= -\dl(x)-W_{\mathrm{tot}}\left(G_x\right)} =n-\frac{1}{2}+\frac{1}{2} \dl(x)+\frac{1}{2} W_{\mathrm{tot}}\left(G_x\right),
    \end{align*}
   where the equality within the underbraces above is obtained from \cref{eq:total_weight_BPSP}.
\end{proof}

\subsection{Reducing BPSP to Weighted MaxCut}

In this section, we present a polynomial-time reduction from the BPSP problem to the weighted maximum cut (MaxCut) problem. We begin by introducing the necessary terminology for defining the MaxCut problem:
Consider a weighted graph $G = (V,E,W)$, where $V = [n]$ for simplicity and $W:E\to\mathbb R$ represents the edge weights. A \textit{cut} of $G$ is defined as a binary string $z \in \{r,b\}^n$. The \textit{weight} of a cut $z \in \{r,b\}^n$ in $G$ is defined as
\begin{align}
\label{eq:weight_cut_def}
    \wt_G(z) = \sum_{\{i,j\} \in E} W(\{i,j\}) [z_i \neq z_j].
\end{align}

The MaxCut problem, whose goal is to identify the cut that maximises the total weight of the cut edges in the graph, can be expressed as follows:

\par\noindent\rule{\textwidth}{0.5pt}
\textbf{\refstepcounter{Problem}Problem~\theProblem\label{problem:MaxCut}.} \textsc{Weighted Maximum Cut (MaxCut)} 
\hfill \nopagebreak\\[-6pt]\nopagebreak
\noindent\rule{\textwidth}{0.3pt} \nopagebreak\\[-10pt]
\begin{tabular}{rrl}
\textbf{Input:}     &  
\multicolumn{2}{l}{a weighted graph $G = (V, E, W)$, where $V = [n]$ and $W:E \to \mathbb R$.} \\
\textbf{Output:}
& maximize & \quad $\wt_G(z)$, as defined in \cref{eq:weight_cut_def}
\\
& subject to & \quad $z \in \{r,b\}^n$.
\end{tabular}
\\
\noindent\rule{\textwidth}{0.5pt}

Also, we denote the optimal cost and the set of optimal solutions of Problem~\ref{problem:MaxCut}, respectively, as follows:
\begin{align}
\textsc{MaxCut}(G) &=  \max_{z \in \{r,b\}^n} \wt_G(z) \\
\textsc{MaxCut}^*(G) &=  \argmax_{z \in \{r,b\}^n} \ \wt_G(z).
\end{align}

Our next proposition provides a formula for computing the weight of a cut of a BPSP graph $G_x$ for a given BPSP instance $x$, expressed in terms of the eta function.

\begin{corollary}\label{cor:cut_weight_eta}
    Let $x \in \Gamma_n$ be a BPSP instance, and denote its BPSP graph as $G_x=\left(V_x, E_x, W_x\right)$. For a cut $z \in \{r, b\}^n$ of $G_x$, the weight of $z$ in $G_x$ is 
    \begin{align}
    \label{eq:weight_cut_BPSP_sum}
        \wt_{G_x}(z)=-\sum_{k=1}^{2 n-1}(-1)^{\eta(x, k)}\left[z_{x_k} \neq z_{x_{k+1}}\right].
    \end{align}
\end{corollary}

\begin{proof}
To begin, we note the similarity between the formulae for the weight of a cut specified in \cref{eq:weight_cut_def} and the total weight of a graph given in \cref{eq:total_weight_def}, where the sole difference is the presence of the indicator function $[z_i \neq z_j]$ in each term of the sum in the former formula. Accordingly, our proof parallels that of Proposition~\ref{prop:total_weight_BPSP}:

\begingroup
\allowdisplaybreaks
    \begin{align*}
    \wt_{G_x}(z) &=\sum_{\{a, b\} \in E_x} W_x(\{a, b\}) \left[z_a \neq z_b\right] \\
    &=\sum_{\{a, b\} \in E_x} \theta_x(\{a, b\}) \left[z_a \neq z_b\right] \\
    &=\sum_{\substack{
    {i \in [2 n-1]} \\ x_i \neq x_{i+1} \\ \theta_x\left(\left\{x_i, x_{i+1}\right\}\right) \neq 0}} \frac{\theta_x\left(\left\{x_i, x_{i+1}\right\}\right) \left[z_{x_i} \neq z_{x_{i+1}}\right]}{\mid\left\{j \in[2 n-1]:\left\{x_j, x_{j+1}\right\}=\left\{x_i, x_{i+1}\right\} \mid\right.} \\
    &=\sum_{\substack{i \in[2 n-1] \\ x_i \neq x_{i+1}}} \frac{\theta_x\left(\left\{x_i, x_{i+1}\right\}\right) \left[z_{x_i} \neq z_{x_{i+1}}\right]}{\mid\left\{j \in[2 n-1]:\left\{x_j, x_{j+1}\right\}=\left\{x_i, x_{i+1}\right\} \mid\right.} \\
    &=\sum_{\substack{
    {
    i \in [2 n-1]} \\ x_i \neq x_{i+1}}} \frac{1}{\mid\left\{
    {j \in [2 n-1]}:\left\{x_j, x_{j+1}\right\}=\left\{x_i, x_{i+1}\right\} \mid\right.}\left\{-\sum_{k=1}^{2 n-1}(-1)^{\eta\left(x, k\right)} \delta_{\left\{x_i, x_{i+1}\right\},\left\{x_k, x_{k+1}\right\}}\right\} \left[z_{x_i} \neq z_{x_{i+1}}\right] \\
    &=-\sum_{k=1}^{2 n-1}(-1)^{\eta(x, k)} \sum_{\substack{i \in[2 n-1] \\ x_i \neq x_{i+1}}} \frac{\delta_{\left\{x_{i}, x_{i+1}\right\},\left\{x_k, x_{k+1}\right\}}}{\mid\left\{j \in[2 n-1]:\left\{x_j, x_{j+1}\right\}=\left\{x_i, x_{i+1}\right\} \mid\right.} \left[z_{x_i} \neq z_{x_{i+1}}\right]\\
    &=-\sum_{k=1}^{2 n-1}(-1)^{\eta(x, k)} \sum_{i \in[2 n-1]} [x_i \neq x_{i+1}] \frac{\delta_{\left\{x_i, x_{i+1}\right\},\left\{x_k, x_{k+1}\right\}}}{\mid\left\{j \in[2 n-1]:\left\{x_j, x_{j+1}\right\}=\left\{x_i, x_{i+1}\right\} \mid\right.} \left[z_{x_i} \neq z_{x_{i+1}}\right]\\
    &=-\sum_{k=1}^{2 n-1}(-1)^{\eta(x, k)} \frac{[x_k \neq x_{k+1}] \left[z_{x_k} \neq z_{x_{k+1}}\right]}{\mid\left\{j \in[2 n-1]:\left\{x_j, x_{j+1}\right\}=\left\{x_k, x_{k+1}\right\} \mid\right.} \sum_{i \in[2 n-1]} \delta_{\left\{x_{i}, x_{i+1}\right\},\left\{x_k, x_{k+1}\right\}} \\
    &=-\sum_{k=1}^{2 n-1}(-1)^{\eta(x, k)} \frac{[x_k \neq x_{k+1}] \left[z_{x_k} \neq z_{x_{k+1}}\right]}{\mid\left\{j \in[2 n-1]:\left\{x_j, x_{j+1}\right\}=\left\{x_k, x_{k+1}\right\} \mid\right.} \mid\left\{i \in[2 n-1]:\left\{x_{i}, x_{i+1}\right\}=\left\{x_k,x_{k+1}\right\} \mid\right. \\
    &=-\sum_{\substack{k=1 \\ x_k \neq x_{k+1}}}^{2 n-1}(-1)^{\eta(x, k)} \left[z_{x_k} \neq z_{x_{k+1}}\right] \\
    &=-\sum_{k=1}^{2 n-1} (-1)^{\eta(x, k)} \left[z_{x_k} \neq z_{x_{k+1}}\right].
\end{align*}
\endgroup
The first line above stems from the definition of the weight of a cut, as provided in \cref{eq:weight_cut_def}. Subsequent lines 2 through 7, and 9 through 10, are derived from the same line of reasoning as the chain of equalities presented in \cref{eq:cal_total_weight}. 

To show that the eighth line follows from the seventh, it suffices to prove that if
\begin{align}
\label{eq:seventh_to_eighth_supposition}
    \{x_i,x_{i+1}\} = \{x_k,x_{k+1}\},
\end{align} then \begin{align}
\label{eq:seventh_to_eighth}
    [x_i \neq x_{i+1}][z_{x_i}\neq z_{x_{i+1}}] = [x_k \neq x_{k+1}][z_{x_k}\neq z_{x_{k+1}}].
\end{align}
To prove this, let us assume that \cref{eq:seventh_to_eighth_supposition} holds. Then the ordered pair $(x_i, x_{i+1})$ is equal to either $(x_k,x_{k+1})$ or 
$(x_{k+1},x_k)$. In either case, $x_i = x_{i+1}$ holds if and only if $x_k = x_{k+1}$ holds. Hence,
\begin{align}
\label{eq:seventh_to_eighth_part1} [x_i \neq x_{i+1}] = [x_k \neq x_{k+1}].
\end{align}
Similarly, \cref{eq:seventh_to_eighth_supposition} implies that the ordered pair $(z_{x_i},z_{x_{i+1}})$ is equal to either
$(z_{x_k},z_{x_{k+1}})$ or $(z_{x_{k+1}},z_{x_k})$. In either case, $z_{x_i} \neq z_{x_{i+1}}$ holds if and only if $z_{x_k} \neq z_{x_{k+1}}$ holds. Hence, 
\begin{align}
\label{eq:seventh_to_eighth_part2}
    [z_{x_i}\neq z_{x_{i+1}}] = [z_{x_k}\neq z_{x_{k+1}}].
\end{align}
Combining \cref{eq:seventh_to_eighth_part1} and \cref{eq:seventh_to_eighth_part2} yields 
\cref{eq:seventh_to_eighth}. This concludes our proof that \cref{eq:seventh_to_eighth_supposition} implies \cref{eq:seventh_to_eighth}, thereby completing the derivation of the eighth line.

Finally, the derivation of the eleventh line follows from the following
observation: when $x_k=x_{k+1}$, we have
$z_{x_k}=z_{x_{k+1}}$, and hence $[z_{x_k}\neq z_{x_{k+1}}]=0$. Therefore, the terms for which $x_k=x_{k+1}$ make no contribution to the
sum. Such terms may consequently be included or omitted without changing
its value, which establishes the equivalence of the tenth and eleventh
lines.
\end{proof}

Our next objective is to establish a formula that connects $\tilde \xi_n(x,z)$---the measure of colour swaps required in colouring the BPSP instance $x$ with the colouring $z$, defined in \cref{eq:xi_tilde}---to the weight of the cut $z$ in the BPSP graph $G_x$. To achieve this, we note that the expression for $\tilde \xi_n(x,z)$ that was derived in \cref{eq:tildexin} contains terms of the form $\left[\neg^a u \neq v\right]$. To simplify these terms, our subsequent lemma establishes an identity that re-expresses them in a more useful form:
\begin{lemma}
\label{lem:useful_equality}
    For any $u, v \in\{r, b\}$ and $a \in\{0,1\}$, the following equality holds: 
    \begin{align}
    \label{eq:useful_equality}
    \left[\neg^a u \neq v\right]=a+(-1)^a[u \neq v].
    \end{align}
\end{lemma}
\begin{proof}
    We analyse the cases $a=0$ and $a=1$ separately, establishing the equivalence between the expressions on the left-hand side (LHS) and the right-hand side (RHS) of  \cref{eq:useful_equality}: When $a = 0$, 
        \begin{align}
            \mathrm{RHS} =0+(-1)^0[u \neq v]=[u \neq v]=\mathrm{LHS}.
        \end{align}
When $a = 1$,
        \begin{align}
        \operatorname{RHS} & =1+(-1)[u \neq v]=1-[u \neq v] 
        =[u=v]=[\neg u \neq v]=\mathrm{LHS}.
        \end{align}
In both cases, we observe that the RHS equals the LHS, thus establishing the validity of the lemma.
\end{proof}

We are now poised to leverage Lemma~\ref{lem:useful_equality} to express the number of colour swaps $\tilde{\xi}_n(x, z)$ in terms of the cut weight
$\wt_{G_x}(z)$. Strikingly, we find that $\tilde{\xi}_n(x, z)$ is precisely the difference between $\xi_n\left(\sigma_{\RF}(x)\right)$---the number of colour swaps in the red-first encoding---and the cut weight $\wt_{G_x}(z)$:

\begin{proposition} \label{cor:paint_swap=red-first_minus_cut weight}
    For any $z \in\{r, b\}^n$ and $x \in \Gamma_n$,
    \begin{align}
    \label{eq:paint_swap=red-first_minus_cut weight}
\tilde{\xi}_n(x, z)=\xi_n\left(\sigma_{\RF}(x)\right)-\wt_{G_x}(z).
\end{align}
\end{proposition}
\begin{proof}
To prove the proposition, we carry out the following calculation:
\begingroup
\allowdisplaybreaks
    \begin{align*}
\tilde{\xi}_n(x, z) & =\sum_{i=1}^{2 n-1}\left[\neg^{\eta(x, i)} z_{x_i} \neq z_{x_{i+1}}\right] \\
& =\sum_{i=1}^{2 n-1}\left\{\eta(x, i)+(-1)^{\eta(x, i)}\left[z_{x_i} \neq z_{x_{i+1}}\right]\right\} \\
& =\sum_{i=1}^{2 n-1} \eta(x, i)+\sum_{i=1}^{2 n-1}(-1)^{\eta(x, i)}\left[z_{x_i} \neq z_{x_{i+1}}\right] \\
& =\xi_n\left(\sigma_{\RF}(x)\right)-\wt_{G_x}(z).
\end{align*}
\endgroup
The first line above corresponds to \cref{eq:tildexin}. The second line follows from utilising Lemma~\ref{lem:useful_equality} to rewrite the indicator function as a sum of two terms. The third line rearranges the second line into a sum of two terms, ensuring that the first term equals $\xi_n\left(\sigma_{\RF}(x)\right)$ by \cref{eq:red-first-cost}, and the second term equals $-\wt_{G_x}(z)$ by \cref{eq:weight_cut_BPSP_sum}. From this rearrangement, we observe that the fourth line ensues.
\end{proof}

Next, note that \cref{eq:red_first_cost_total_weight} provides an equivalent expression for the number of colour swaps $\xi_n\left(\sigma_{\RF}(x)\right)$ in the red-first encoding. By substituting \cref{eq:red_first_cost_total_weight} into \cref{eq:paint_swap=red-first_minus_cut weight}, we obtain the following corollary:
\begin{corollary}\label{cor:xinxz}
    For any $z \in\{r, b\}^n$ and $x \in \Gamma_n$,
    \begin{align}
    \tilde{\xi}_n(x,z)=n-\frac{1}{2}+\frac{1}{2} \dl(x)+\frac{1}{2} W_{\mathrm{tot}}\left(G_x\right)-\wt_{G_x}(z).
    \label{eq:xinxz}
    \end{align}
\end{corollary}
The above equation expresses the number of colour swaps involved when $x$ is coloured with $z$ in terms of the double-letter count of $x$, the total weight of the BPSP graph $G_x$, and the weight of the cut $z$ in $G_x$.

Next, by optimising \cref{eq:xinxz} over the cuts $z\in \{r,b\}^n$, we obtain the following corollary:  

\begin{corollary}
\label{cor:BPSP_reduces_to_MaxCut}
    Let $x \in \Gamma_n$. Then, we have that
    \begin{align}
    \label{eq:BPSP_maxcut}
        \textsc{BPSP}(x) &=\xi_n\left(\sigma_{\mathrm{RF}}(x)\right)-\textsc{MaxCut}(G_x) = n-\frac{1}{2}+\frac{1}{2} \dl(x)+\frac{1}{2} W_{\mathrm{tot}}\left(G_x\right)-\textsc{MaxCut}\left(G_x\right),
        \\
        \label{eq:BPSP_maxcut*}
        \widetilde{\textsc{BPSP}^*}(x)&=\textsc{MaxCut}^*\left(G_x\right).
    \end{align}
\end{corollary}

\begin{proof}
    Let $x \in \Gamma_n \subset \Gamma$. Then \cref{eq:BPSP_maxcut} holds because
    \begingroup
\allowdisplaybreaks
    \begin{align*}
    \textsc{BPSP}(x) & =\min_{z \in \{r, b\}^n} \tilde{\xi}_n(x, z) 
    =\min_{z \in\{r, b\}^n} \left[\xi_n(\sigma_{\mathrm{RF}}(x))- \wt_{G_x}(z)\right] \\ &=\xi_n(\sigma_{\mathrm{RF}}(x))-\max_{z \in\{r, b\}^n} \wt_{G_x}(z) \\ &=\xi_n(\sigma_{\mathrm{RF}}(x))-\textsc{MaxCut}(G_x) = n-\frac{1}{2}+\frac{1}{2} \#_{d l}(x)+\frac{1}{2} W_{\mathrm{tot}}\left(G_x\right)-\textsc{MaxCut}\left(G_x\right),
    \end{align*}
    \endgroup
    where the last equality follows from \cref{eq:red_first_cost_total_weight}. Similarly, \cref{eq:BPSP_maxcut*} holds because
 \begingroup
\allowdisplaybreaks
     \begin{align*}
    \widetilde{\textsc{BPSP}^*}(x) & =\argmin_{z \in \{r, b\}^n} \tilde{\xi}_n(x, z) 
    =\argmin_{z \in\{r, b\}^n} \left[\xi_n(\sigma_{\mathrm{RF}}(x))- \wt_{G_x}(z)\right] \\ &=\argmin_{z \in\{r, b\}^n} \left(-\wt_{G_x}(z)\right) = \argmax_{z \in\{r, b\}^n} \left(\wt_{G_x}(z)\right) =\textsc{MaxCut}^*(G_x).
    \end{align*}
    \endgroup
    This completes the proof of the corollary.
\end{proof}

Finally, we note that Corollary~\ref{cor:BPSP_reduces_to_MaxCut} provides us with a polynomial-time reduction from the BPSP problem stated in Problem~\ref{problem:BPSP} to the weighted MaxCut problem presented in Problem~\ref{problem:MaxCut}. Specifically, one use of an oracle for Problem~\ref{problem:MaxCut} enables us to solve Problem~\ref{problem:BPSP} in polynomial time, as we demonstrate in Algorithm~\ref{alg:reduction_BPSP_MaxCut}.

\begin{algorithm}
\DontPrintSemicolon
\KwIn{$x \in \Gamma$}
\textbf{Output:} $\mathrm{BPSP}(x)$ and an optimal colouring
$f^*\in\mathrm{BPSP}^*(x)$.

Construct $G_x$, the \textsc{BPSP} graph of $x$, utilising Definition~\ref{def:bpsp_graph}.\;
Execute the weighted $\textsc{MaxCut}$ oracle for Problem~\ref{problem:MaxCut} on $G_x$ to obtain $\textsc{MaxCut}(G_x)$ and $z^* \in \textsc{MaxCut}^*(G_x)$.\;
Compute $\xi_n(\sigma_{\mathrm{RF}}(x))$ using \cref{eq:red_first_cost_total_weight}.\;
Compute 
$f^* = \mathcal E(x,z^*)$ using \cref{eq:En_function}.
\;
\textbf{return} $\mathrm{BPSP}(x)
=\xi_n(\sigma_{\mathrm{RF}}(x))-\mathrm{MaxCut}(G_x)$
and $f^*\in\mathrm{BPSP}^*(x)$.
\caption{Polynomial-time reduction from \textsc{BPSP} (Problem~\ref{problem:BPSP}) to \textsc{MaxCut}
(Problem~\ref{problem:MaxCut})
}
\label{alg:reduction_BPSP_MaxCut}
\end{algorithm}

\section{Ising Formulation of BPSP}
\label{sec:Ising_formulation_BPSP}

In the Ising formulation, we assign numerical values to the formal symbols $r$ and $b$: 
\begin{align}
    r \equiv 1, \qquad b \equiv -1.
\end{align}

In the following lemma, we present several identities that will prove useful for translating formulas in \cref{sec:BPSP_graph} into an algebraic form:

\begin{lemma}
\label{lem:algebraic_identities}
    Let $y, z \in \{r, b\}$, where $r=1$ and $b=-1$, and $a \in \{0,1\}$. Then, we have:
    \begin{enumerate}[label=\normalfont{(\roman*)}]
        \item $\neg z=-z$.
        \item $\neg^a z=(-1)^a z$.
        \item $[y=z]=\frac{1}{2}(1+y z)$.
        \item $[y \neq z]=\frac{1}{2}(1-y z)$.
        \item \label{item:negazy} $\left[\neg^a z \neq y\right]=\frac{1}{2}-\frac{1}{2}(-1)^a z y$.
    \end{enumerate}
\end{lemma}
\begin{proof}  To prove these identities, it suffices to check that equality holds for each of the possible discrete binary values of $y$, $z$, and $a$. Indeed, (i) holds since $\neg r = b = -1 = -r$ and $\neg b = r = 1 = -b$. (ii) holds because $a=0$ implies that $\neg^a z=z=(-1)^0 z$ and $a=1$ implies that $\neg^a z=\neg z=-z=(-1)^1 z$. (iii) holds because
\begin{align}
    yz= \begin{cases}
1, & y=z \\ -1, & y \neq z
\end{cases}\quad
\implies \quad
    \frac{1}{2}(1+yz) = \begin{cases}
1, & y=z \\ 0, & y \neq z
\end{cases}\
= \  [y=z].
\end{align}
(iv) holds since $[y \neq z]=1-[y=z]=1-\frac{1}{2}(1+y z)=\frac{1}{2}(1-y z)$. Finally, (v) holds because ${\left[\neg^a z \neq y\right] }  =\left[(-1)^a z \neq y\right] =\frac{1}{2}\left(1-(-1)^a z y\right) = \frac{1}{2}-\frac{1}{2}(-1)^a z y$.
\end{proof}

\subsection{Ising Expressions for BPSP Cost Function}

We now present a formula for the colour swap tally $\tilde \xi_n(x,z)$ from \cref{eq:tildexin} in terms of the Ising formulation.

\begin{proposition} \label{prop:Ising1}
    Consider $z \in\{-1,1\}^n$ and $x \in \Gamma_n$. Then, we can express the number of colour swaps required when $x$ is coloured with $z$ as
    \begin{align} \label{eq:xi_Ising1}
\tilde{\xi}_n(x, z)&=n-\frac{1}{2}-\frac{1}{2} \sum_{i=1}^{2 n-1}(-1)^{\eta(x, i)} z_{x_i} z_{x_{i+1}}
\\
\label{eq:xi_Ising2}
&=n-\frac{1}{2}+\frac{1}{2} \dl(x)-\frac{1}{2} \sum_{\substack{i=1 \\ x_i \neq x_{i+1}}}^{2 n-1}(-1)^{\eta(x, i)} z_{x_i} z_{x_{i+1}}.
\end{align}
\end{proposition}
\begin{proof} 
By applying Lemma~\ref{lem:algebraic_identities}\ref{item:negazy} to \cref{eq:tildexin} and simplifying, we obtain:
\begingroup
\allowdisplaybreaks
    \begin{align*}
    \widetilde{\xi}_n(x, z) & =\sum_{i=1}^{2 n-1}\left[\neg^{\eta(x, i)} z_{x_i} \neq z_{x_{i+1}}\right] \\
    & =\sum_{i=1}^{2 n-1} \frac{1}{2}\left(1-(-1)^{\eta(x, i)} z_{x_i} z_{x_{i+1}}\right) \\
    & =n-\frac{1}{2}-\frac{1}{2} \sum_{i=1}^{2 n-1}(-1)^{\eta(x, i)} z_{x_i} z_{x_{i+1}} \\
    &= n-\frac{1}{2}-\frac{1}{2}\left(\sum_{\substack{i=1 \\
    x_i=x_{i+1}}}^{2 n-1}+\sum_{\substack{i = 1 \\
    x_i \neq x_{i+1}}}^{2 n-1}\right)(-1)^{\eta(x, i)} z_{x_i} z_{x_{i+1}} \\
    &=n-\frac{1}{2}-\frac{1}{2} \sum_{\substack{i=1 \\ x_i \neq x_{i+1}}}^{2 n-1}(-1)^{\eta(x, i)} z_{x_i} z_{x_{i+1}}-\frac{1}{2} \sum_{\substack{i=1 \\ x_i=x_{i+1}}}^{2 n-1}(-1)^{\eta(x, i)} z_{x_i} z_{x_{i+1}}\\
    &=n-\frac{1}{2}-\frac{1}{2} \sum_{\substack{i=1 \\ x_i \neq x_{i+1}}}^{2 n-1}(-1)^{\eta\left(x,i\right)} z_{x_i} z_{x_{i+1}}+\frac{1}{2}\left|\left\{i \in[2 n-1]: x_i=x_{i+1}\right\}\right| \\
    &=n-\frac{1}{2}-\frac{1}{2} \sum_{\substack{i=1 \\ x_i \neq x_{i+1}}}^{2 n-1}(-1)^{\eta(x, i)} z_{x_i} z_{x_{i+1}}+\frac{1}{2} \dl(x),
    \end{align*}
    \endgroup
where the sixth line follows from the observation that when $x_i = x_{i+1}$, we have $\eta(x,i) = 1$ (by Lemma~\ref{eq:eta_consecutive_symbols}) and $z_{x_i}z_{x_{i+1}} = z_{x_i}^2 = 1$. The third and last lines provide the desired expressions \cref{eq:xi_Ising1} and \cref{eq:xi_Ising2}, respectively.
\end{proof}

Finally, through optimisation of \cref{eq:xi_Ising1} and  \cref{eq:xi_Ising2} over the cuts $z\in \{-1,1\}^n$, and noting that for $n=|x|/2$ for $x \in \Gamma_n$, we arrive at the following corollary:  

\begin{corollary}  \label{cor:Ising2}
    Let $x \in \Gamma$ be a BPSP instance. Then the optimal cost and the set of optimal solutions of the BPSP problem are, respectively,
\begingroup
\allowdisplaybreaks
    \begin{align}
    \textsc{BPSP}(x) & =\frac{1}{2}\left\{|x|-1-\max_{z \in \{-1,1\}^{|x|/2}} \sum_{i=1}^{|x|-1}(-1)^{\eta(x, i)} z_{x_i} z_{x_{i+1}}\right\} \\
    & =\frac{1}{2}\left\{|x|-1+\dl(x)-\max_{z \in \{-1,1\}^{|x|/2}} \sum_{\substack{i=1 \\ x_i \neq x_{i+1}}}^{|x|-1} (-1)^{\eta(x, i)} z_{x_i} z_{x_{i+1}}\right\}
\\    \widetilde{\textsc{BPSP}^*}(x) & =\argmax_{z \in\{-1,1\}^{|x| / 2}} \sum_{i=1}^{|x|-1}(-1)^{\eta(x, i)} z_{x_i} z_{x_{i+1}} \\
& =\argmax_{z \in\{-1,1\}^{|x| / 2}} \sum_{\substack{i=1 \\
x_i \neq x_{i+1}}}^{|x|-1}(-1)^{\eta(x, i)} z_{x_i} z_{x_{i+1}}.
\end{align}
\endgroup
\end{corollary}

In essence, these expressions unveil the underlying connection between solving the BPSP and finding the ground state energy of the associated Ising Hamiltonian:
\begin{align}
    H(z_1,\ldots, z_n) = -\sum_{\substack{i=1 \\
x_i \neq x_{i+1}}}^{2n-1}(-1)^{\eta(x, i)} z_{x_i} z_{x_{i+1}}.
\end{align}

\newpage
\twocolumngrid
\bibliography{refs}

\end{document}